\documentclass[
final
]{dmtcs-episciences}

\usepackage[utf8]{inputenc}
\usepackage{subfigure}

\usepackage{pgf,tikz}
\usepackage{amsmath,amssymb,amsthm,amscd,amsfonts}
\usepackage{algorithm}	
\usepackage{algpseudocode}
\usepackage{multicol}
\newcommand*{\mybox}[1]{\framebox{#1}}

\usepackage[round]{natbib}

\newtheorem{conjecture}{ Conjecture}[section]
\newtheorem{theorem}[conjecture]{ Theorem}
\newtheorem{lemma}[conjecture]{ Lemma}
\newtheorem{corollary}[conjecture]{ Corollary}
\newtheorem{proposition}[conjecture]{ Proposition}
\newtheorem{remark}[conjecture]{ Remark}

\newtheorem{definition}[conjecture]{ Definition}

\newtheorem{example}[conjecture]{ Example}

\allowdisplaybreaks

\author[Anuran Maity et. al]{Kalpana Mahalingam \affiliationmark{1}
  \and Anuran Maity\affiliationmark{1,2,3}\thanks{Corresponding author}}
\title[Watson-Crick conjugates of words and languages]{Watson-Crick Conjugates of Words and Languages}

\affiliation{
  Department of Mathematics, Indian Institute of Technology Madras, Chennai, India\\
  Saint Petersburg State University, 7/9 Universitetskaya nab., St. Petersburg, Russia\\
  Department of Mathematics, Indian Institute of Technology Guwahati, Guwahati, India}
  

\begin{document}

\publicationdata{vol. 27:3}{2025}{10}{10.46298/dmtcs.13593}{2024-05-15; 2024-05-15; 2025-03-24; 2025-08-08}{2025-08-19}

\maketitle

\begin{abstract}
  In this work, we explore the concept of Watson-Crick conjugates, also known as $\theta$-conjugates (where $\theta$ is an antimorphic involution), of words and languages. 
 This concept extends the classical idea of conjugates by incorporating the Watson-Crick complementarity of DNA sequences.
Our investigation initially focuses on the properties of $\theta$-conjugates of words. We then define $\theta$-conjugates of a language and study closure properties of certain families of languages under the $\theta$-conjugate operation. Furthermore, we analyze the iterated $\theta$-conjugate of both words and languages. 
Finally, we discuss the idea of $\theta$-conjugate-free languages and examine some decidability problems related to it.
\end{abstract}



\section{Introduction}

Conjugates of a word are significant in the field of combinatorics on words as they have numerous applications in cryptography (\cite{belkheyar2023introducing}), biology (\cite{walter2015consistent}), and other areas. 
Our analysis in this work focuses on various combinatorial, structural, and language-theoretic properties of Watson-Crick conjugates, or $\theta$-conjugates (where $\theta$ is an antimorphic involution) of words  - a generalization of the conjugates of words.

Over the years, DNA sequences have been used to encode/design  a range of complex problems/experim-ents (for example, see \cite{adleman1994molecular}, \cite{aoi1998solution}, \cite{bi2016initiator}, \cite{braich2002solution}, \cite{head2000computing}, \cite{henkel2007dna}, \cite{henkel2005protein}, \cite{jeddi2017three}, \cite{johnson2008automating}, \cite{ouyang1997dna}, \cite{sakamoto2000molecular}, \cite{shi2017recent}, \cite{wang2008solving}, \cite{zhang2011dynamic}, etc). 
When encoding or designing a problem using DNA sequences, the occurrence of inter/intra-molecular hybridizations (see Fig. \ref{intera}) between DNA strands can be perceived both as an advantage and a disadvantage. For example, formation of hairpin structures have found utility in various applications, including solving the satisfiability problem (\cite{sakamoto2000molecular}), designing DNA aptamers (\cite{jeddi2017three}), and programming catalytic DNA self-assembly strategies (\cite{bi2016initiator}).
Conversely, these hybridizations are sometimes viewed as a disadvantage (see \cite{shi2017recent,zhang2011dynamic}).
This is primarily because these hybridizations can significantly disrupt the intended interactions between the involved sequences  in inter/intra-molecular hybridizations  and other DNA strands in the  preprogrammed and expected ways.   Therefore, to  avoid such undesirable hybridizations, careful DNA sequence design is necessary at the time of encoding  the problem.  Drawing inspiration from this approach of DNA sequence design, several authors have proposed various conditions for encoding
 by extending various concepts of classical combinatorics on words, such as {primitive words} (\cite{Ehsan}, \cite{CZ12},  \cite{gawrychowski2014testing},
\cite{kari2014generating}, \cite{kari2019state}),   {bordered words}
(\cite{kari2014pseudo},  \cite{kari2017disjunctivity}, \cite{kari2007involutively}, \cite{kari2008watson}, \cite{watson}, 
\cite{kari2009pseudoknot}),
{palindromes} (\cite{watson}, \cite{LKari2010}), and {rich words}
(\cite{starosta2011theta}).  
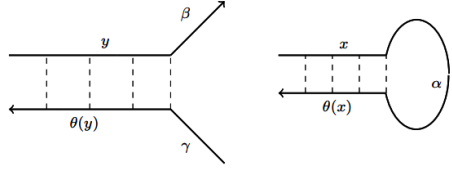
\begin{figure}[h]
\centering
\begin{tikzpicture}
\draw[-] (-3,0) -- (0,0) node[midway, above]{$y$};
\draw[<-] (-3,-1) -- (0,-1) node[midway, below]{$\theta_{WC}(y)$};
\draw[dashed] (-2.5,0) -- (-2.5,-1);
\draw[dashed] (-2,0) -- (-2,-1);
\draw[dashed] (-1.5,0) -- (-1.5,-1);
\draw[dashed] (-1,0) -- (-1,-1);
\draw[dashed] (-0.5,0) -- (-0.5,-1);
\draw[dashed] (0,0) -- (0,-1);
\draw[->] (0,0) -- (1,1) node[midway, above]{$\beta$};
\draw[thick] (0,-1) -- (1,-2) node[midway, below]{$\gamma$};
\draw[-] (3,0) -- (6,0) node[midway, above]{$x$};
\draw[<-] (3,-1) -- (6,-1) node[midway, below]{$\theta_{WC}(x)$} ;
\draw[dashed] (3.5,0) -- (3.5,-1);
\draw[dashed] (4,0) -- (4,-1);
\draw[dashed] (4.5,0) -- (4.5,-1);
\draw[dashed] (5,0) -- (5,-1);
\draw[dashed] (5.5,0) -- (5.5,-1);
\draw[dashed] (6,0) -- (6,-1);
\draw (6,0) .. controls (7.5,2.5) and (7.5,-3.5) .. (6,-1) node[midway, above left] {$\alpha$};
\end{tikzpicture}
\caption{ Example of inter- and intra- molecular structures }
\label{intera}
\end{figure}

In \cite{kari2007involutively}, the authors extended the concept of bordered word to $\theta$-bordered word from the perspective of DNA computing, where $\theta$ is an antimorphic involution (i.e., $\theta(xy)=\theta(y) \theta(x)$ and $\theta^2(t)=t$ for any $x, y, t \in \Sigma^*$) that incorporates the notion of the Watson-Crick complementarity of DNA molecules. 
They also defined $\theta$-bordered word for a morphic involution $\theta$ (i.e., $\theta(xy)=\theta(x) \theta(y)$ and $\theta^2(t)=t$ for any $x, y, t \in \Sigma^*$) as a generalization of bordered word.
 For a morphic/antimorphic involution $\theta$, a word $w$ is said to be $\theta$-bordered if $w=uz=\theta(z)v$ for some non-empty words $u, z, v$. A word that is not $\theta$-bordered is called $\theta$-unbordered.
 When focusing on the DNA alphabet $\{A, T, G, C \}$ and utilizing the Watson-Crick complementarity function $\theta_{WC}$ ($\theta_{WC}$ is an antimorphic involution on $\{A, T, G, C\}^*$ such that $\theta_{WC}(A)=T$, $\theta_{WC}(T)=A$, $\theta_{WC}(C)=G$, $\theta_{WC}(G)=C$),  $\theta_{WC}$-unbordered words represents the collection of DNA single strands that are free from both inter- and intra-molecular structures such as the ones shown in Fig. \ref{intera}(left), and hairpins shown in Fig. \ref{intera}(right). 
 From the definition of $\theta$-bordered word,
 \cite{watson} defined the notion of $\theta$-conjugate of a word for morphic/antimorphic involution $\theta$.
Formally, a word $u$ is called a $\theta$-conjugate of a word $v$ if there exists a word $z$ such that $u z = \theta(z)v$. The set of all $\theta$-conjugates of $v$ is denoted by $C_\theta(v)$. 
When $\theta$ is an antimorphic involution over $\Sigma^*$, then analyzing the solution of the equation $u z = \theta(z)v$ (as presented in Proposition $2$ in \cite{watson}), it becomes evident that for $v = xy$, the $\theta$-conjugate $u$ can be expressed as $\theta(y)x$ for some $x, y \in \Sigma^*$ (see Fig. \ref{fig1}).
\begin{figure}
    \centering
    \includegraphics[width=.7\textwidth]{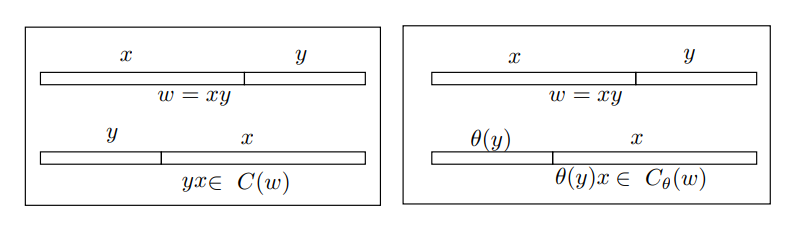}
     \caption{Structure of conjugates $\&$ $\theta$-conjugates of word for antimorphic involution $\theta$}
    \label{fig1}
\end{figure}
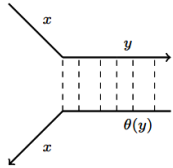
\begin{figure}[h]
\centering
\begin{tikzpicture}
\draw[-] (-4,1) -- (-3,0) node[midway, above]{$x$};
\draw[->] (-3,0) -- (0,0) node[midway, above]{$y$};
\draw[-] (-3,-1) -- (0,-1) node[midway, below]{$\theta_{WC}(y)$};
\draw[<-] (-4,-2) -- (-3,-1) node[midway, below]{$x$};
\draw[dashed] (-3,0) -- (-3,-1);
\draw[dashed] (-2.5,0) -- (-2.5,-1);
\draw[dashed] (-2,0) -- (-2,-1);
\draw[dashed] (-1.5,0) -- (-1.5,-1);
\draw[dashed] (-1,0) -- (-1,-1);
\draw[dashed] (-0.5,0) -- (-0.5,-1);
\end{tikzpicture}
\caption{ Inter-molecular structure between a word and its $\theta_{WC}$-conjugate }
    \label{interawcbind}
\end{figure}
In this scenario,  if we consider the DNA alphabet $\{A, T, G, C\}$, the Watson-Crick complementarity function $\theta_{WC}$, and employ $u, v \in \{A, T, G, C\}^*$  for encoding a problem, then a possible intermolecular hybridization  (see Fig. \ref{interawcbind}) may occur between $v$ and $u$.
To avoid such a structure at the time of encoding, if  we use $v$, then we should refrain from utilizing any of its $\theta_{WC}$-conjugates that are different from the word $v$.
In addition to being relevant for DNA computing, the notion of $\theta$-conjugates of a word
is a generalization of conjugates of a word, a classical and well studied concept in combinatorics on words.
These motivate us to explore various
 properties of $\theta$-conjugates of a word and $\theta$-conjugate-free languages when $\theta$ is an antimorphic involution.

 \cite{thetapalin2021} discussed several properties of the set of all $\theta$-conjugates of a word when $\theta$ is an antimorphic involution. They mainly studied the cardinality of the set of all $\theta$-conjugates of a word and occurrences of
 palindromes, $\theta$-palindromes in the set of all $\theta$-conjugates of a word.
In this work, we further explore properties of the set of all $\theta$-conjugates of a word when $\theta$ is an antimorphic involution. We  also extend the idea of $\theta$-conjugates of a word to $\theta$-conjugates of a language $L$ as $C_\theta(L) = \bigcup\limits_{w \in L} C_\theta(w)$  and study various properties of $C_\theta(L)$ when $\theta$ is an antimorphic involution. The text is organized as follows. In Section \ref{sec3}, we study some properties of the set of all $\theta$-conjugates of a word. We first investigate the structure of $\theta$-conjugates of a word. Then, we discuss the solution of the equation $C_\theta(u)=C_\theta(v)$.
Section \ref{sec4} investigates the Chomsky hierarchy for $C_\theta(L)$ when 
$L$ is regular, context-free, context-sensitive and recursively enumerable.
In Section \ref{sec5}, we define and discuss the iterated $\theta$-conjugate of words and languages.
Section \ref{sec6} investigates some decidability problems for  $\theta$-conjugate-free languages.  
We end our work with some concluding remarks.

\section{Basic definitions and notations}\label{sec2}
An \textit{alphabet} $\Sigma$ is a finite non-empty set of letters.
A \textit{word} over $\Sigma$ is a sequence of letters from $\Sigma$.  
The \textit{length of a word} $w$, denoted by $|w|$, is the number of letters in $w$.
A word with length zero is called the \textit{empty word} and is denoted by $\lambda$. 
The set of all finite words over $\Sigma$ is denoted by $\Sigma^*$, and $\Sigma^+=\Sigma^*\setminus \{\lambda\}$. 
A \textit{language} $L$ over $\Sigma$ is a subset of $\Sigma^*$.
The set of all words over $\Sigma$ of length $n$ is denoted by $\Sigma^n$ and the set of all words over $\Sigma$ of length greater than or equal to $n$ is denoted by $\Sigma^{\geq n}$.
Let $\textit{Alph}(w)$ denote the set of all letters that occur in a word $w$ and $\textit{Alph}(L)=\bigcup\limits_{w \in L} \textit{Alph}(w)$ where $L$ is a language.
A word $u$ is a \textit{factor} of a word $w$ if $w = xuy$ for some $x, y\in \Sigma^*$. When $x = \lambda$ (respectively $y=\lambda$), $u$ is called a \textit{prefix} (respectively \textit{suffix}) of $w$.
If $u$ is a prefix (respectively suffix) of $w$ and $u \neq w$, then $u$ is called a \textit{proper prefix} (respectively \textit{proper suffix}) of $w$. 
Set of all factors (respectively prefixes and suffixes) of a word $w$ is denoted by $Fac(w)$ (respectively $Pref(w)$ and $Suf(w)$).
For $x \in \Sigma^*$, $|w|_x$ denotes the number of occurrences of $x$ in $w$.
For a word $w= w_1w_2\cdots w_n$ where each $w_i\in \Sigma$, the \textit{reversal} of $w$ is $w^R = w_n\cdots w_2w_1$.
A word $w$ is a \textit{palindrome} if $w=w^R$. 
A word $u$ is a \textit{conjugate} of a word $w$ if $u=yx$ and $w=xy$ for some $x,y\in \Sigma^*$ (see Figure \ref{fig1}).
The set of all conjugates of $w$ is denoted by $C(w)$. 
A word $w \in \Sigma^+$ is called \textit{primitive} if  $w = x^n$ implies 
$w=x$ and  $n=1$. 
For every $w \in \Sigma^+$, there exists a unique $\rho(w)\in \Sigma^+$, called  the \textit{root} of $w$,  such that  $\rho(w)$ is a primitive word and $w = \rho(w)^n $ for some integer $n \geq 1$.

A function  $\theta:\Sigma^{*} \rightarrow \Sigma^{*}$ is said to be an \textit{antimorphism} {(respectively \textit{morphism})} if $\theta(uv)=\theta(v) \theta(u)$ (respectively $\theta(uv)=\theta(u) \theta(v)$) for any $u, v \in \Sigma^*$. 
The function $\theta$ is called \textit{involution} if $\theta^{2}$ is {the} identity mapping on $\Sigma^*$.
For a morphic/antimorphic involution $\theta$, a word $z$ is called a \textit{$\theta$-palindrome} if $z = \theta(z)$.   
For a morphic/antimorphic involution $\theta$, a word $u$ is a \textit{$\theta$-conjugate} of a word $w$ if
 $uv = \theta(v)w$ for some $v \in \Sigma^*$.
 The set of all $\theta$-conjugates of $w$ is denoted by $C_{\theta}(w)$. 
Note that when $\theta$ is an antimorphic involution on $\{A, T, G, C\}^*$ such that $\theta(A)=T$, $\theta(T)=A$, $\theta(C)=G$, $\theta(G)=C$, then $\theta = \theta_{WC}$, i.e., $\theta$ is the \textit{Watson-Crick complementarity function}. In this case $C_\theta(w)$ is called the set of all \textit{Watson-Crick conjugates} of $w$. 

\textbf{In the rest of the text, we consider $\theta$ as an antimorphic involution over an alphabet $\Sigma^*$.}

A \textit{grammar} $G$ is formally defined as a quadruple $(V, T, P, S)$ where $V$ is a finite set of non-terminal symbols, $T$ is a finite set of terminal symbols,  $P$ is a finite set of production rules, and $S \in V$ is the start symbol.
Using production rules of $P$, if a word $w \in T^*$ is derived from $S$, then we say $w$ is generated by the grammar $G$, and it is denoted by  $S\rightarrow^* w$. The set of all words generated by the grammar $G$ is denoted by $L(G)$. 
A grammar $G$ is said to be \textit{context-free} if all the production rules of $P$ are of the form $A \rightarrow u$, where $A \in V$ and $u \in (V \cup T)^*$.
A grammar $G$ is said to be \textit{context-sensitive} if all the production rules of $P$ are of the form $\alpha A \beta \rightarrow \alpha \gamma \beta$, where $A\in V$, $\alpha, \beta \in (V \cup T)^*$ and $\gamma \in (V \cup T)^+$.
A grammar $G$ is said to be \textit{unrestricted} if all the production rules of $P$ are of the form $u \rightarrow v$, where $ u, v \in (V \cup T)^*$ with $u$ containing at least one non-terminal symbol.

A \textit{deterministic finite automaton }(DFA) $M$ is defined by $M = ( S', \Sigma, \delta, q_0, F )$, where $S'$ is a finite set of states, $\Sigma$ is an alphabet, $\delta : S' \times \Sigma \rightarrow S'$ is a transition function, $q_0 \in S'$ is the initial state and $F \subseteq S'$ is the set of all final states.
A word is accepted by the DFA  $M$ if, starting from the initial state $q_0$, processing the word using the transition function $\delta$ leads to a state in $F$. The set of all words accepted by $M$ is denoted by $L(M)$. A language $L_1$ is said to be \textit{regular} if there exists a DFA $M_1$ such that $L_1=L(M_1)$.
{A \textit{linear bounded automaton} is a non-deterministic Turing machine $M' = (Q', \Sigma, \Gamma, \delta, q_0, <, >, B, F)$ where
$Q'$ is the set of all states,  $\Gamma$ is the tape alphabet, $\Sigma \subseteq \Gamma \setminus \{B\}$  is the input alphabet, $\delta: Q' \times \Gamma \rightarrow 2^{Q' \times \Gamma \times \{L, R, N\}}$ is the transition function,
$B \in \Gamma$ is the blank symbol,
$q_0 \in Q'$ is the initial state, $F \subseteq Q'$ is the set of final states,
$< \in \Sigma$ is the left end marker and $> \in \Sigma$ is the right end marker. 
Here, $L, R$ and $N$ denote the possible movements of the tape head: left, right, and no movement, respectively, as dictated by the transition function $\delta$.
A word $w$ is accepted by the linear bounded automaton $M'$ if the automaton, starting in its initial state $q_0$ with the input word $<w>$ on the tape, can transition through a sequence of states according to its transition function $\delta$ and eventually reach one of its final states in $F$.
Here note that if $w$ is a word accepted by a linear bounded automaton $M'$, then $w$ is in $(\Sigma \setminus \{ <, >\})^*$. 
For a state $q$ and two strings $u$, $v$ over the tape alphabet $\Gamma$, we write $u q v$ for the \textit{configuration} of the Turing machine $M'$ where the current state is $q$, the current tape content is $uv$, and the current head location is the first symbol of $v$.
A move of the Turing machine $M'$ from one configuration to another is denoted by $\vdash_{M'}$. 
The symbol ${\vdash}^*_{M'}$ represents an arbitrary number of moves of $M'$.

For all other concepts in formal language theory and combinatorics on words, the reader is referred to \cite{linz2012introduction,Lothaire1997,rozen1997,shallit2008}.

 We recall the following results from the literature, which will be used throughout our work. 
\begin{lemma} \label{4lk}\cite{Schutz62}
Let $u,v,w \in \Sigma^{+}$.
\begin{itemize}
     \item If $uv=vu$, then $u$ and $v$ are powers of a common primitive word. 
     \item If $uv=vw$, then  $u=xy$, $v=(xy)^{k}x$, $w=yx$ for some $k \geq 0$, $x \in \Sigma^{+}$ and $y \in \Sigma^{*}$.
 \end{itemize}
\end{lemma}

\begin{proposition}\label{gl1}\cite{watson}
For $u, v \in \Sigma^+$, if $uv = \theta(v)u$ and
 $\theta$ is an antimorphic involution, then $u = x(yx)^i$, $v = yx$ for some integer $i \geq 0$ and $\theta$-palindromes $x\in \Sigma^+$, $y \in \Sigma^*$.
\end{proposition}

\section{Some Important Properties of set of all $\theta$-conjugates of a Word}\label{sec3}
 
In this section, we investigate several combinatorial properties of the set of all $\theta$-conjugates of a word. 
First, we examine the structure of elements of $C_\theta(w)$ for $w \in \Sigma^*$. This analysis will be useful later on. Next, we show that if  $C_\theta(w)$ contains $|w|+1$ elements, then  $C_\theta(w)$ must contain a primitive word.  
Finally, we discuss the structure of words $u$ and $v$ such that $C_\theta(u)=C_\theta(v)$.

To discuss the structure of elements of $C_\theta(w)$ for some $w$, we first recall the following result from \cite{watson}.
 \begin{proposition}\cite{watson}\label{mainprop}
For an antimorphic involution $\theta$, consider $u$ is a $\theta$-conjugate of $w$, i.e., $uv=\theta(v)w$ for some $v \in \Sigma^*$. Then, either $u = xy$, {$v=\theta(x)$} and $w = y\theta(x)$ for some $x, y \in  \Sigma^*$, or $w = \theta(u)$ and $v=\beta w$ for some $\theta$-palindrome $\beta$.
\end{proposition}

Using Proposition \ref{mainprop}, we obtain the following result, which will be used throughout this text.
\begin{proposition}\label{deftc}
    For a word $w$,
     $C_{\theta}(w) = \{ \theta(y)x:~ w=xy \text{ where } ~x,y \in \Sigma^*\}$.
\end{proposition}
\begin{proof}
    Consider $T=\{ \theta(y)x~:~ w=xy, ~x,y \in \Sigma^*\}$. If $u \in C_\theta(w)$, then
    $uv=\theta(v)w$ for some $v \in \Sigma^*$. Then by Proposition \ref{mainprop}, $u = \theta(y)x$ and $w = xy$ for some $x, y \in  \Sigma^*$ (see Figure \ref{fig1}).
    This implies $u \in T$, i.e., $C_{\theta}(w) \subseteq T$.
    Consider $z \in T$. Then, $z=\theta(y_1) x_1$ where $w=x_1 y_1$ for some $x_1, y_1 \in \Sigma^*$. 
    If $x_1=\lambda$, then for $v_1=\beta y_1$ where $\beta$ is a $\theta$-palindrome, we have $z v_1= \theta(y_1) \beta y_1 =\theta(v_1) w  $, i.e., $z \in C_\theta(w)$. 
    If $x_1 \in \Sigma^+$, then for $v_2= y_1$, we have $z v_2= \theta(y_1) x_1 y_1 =\theta(v_2) w$, i.e., $z \in C_\theta(w)$. Thus, $T \subseteq C_\theta(w)$.
\end{proof}

In the following, we illustrate Proposition \ref{deftc} with several examples. These examples also highlight an important property of the $\theta$-conjugacy relation.
Similar to the conjugacy relation, we define the $\theta$-conjugacy relation $R_\theta$ on $\Sigma^*$, where $(u, v) \in R_\theta$ indicates that $u$ is a $\theta$-conjugate of $v$. It is known that the conjugacy relation is an equivalence relation \cite{Lothaire1997}. However, with the help of the following examples, we show that the $\theta$-conjugacy relation is not an equivalence relation. 

\begin{example}\label{o1uyt}
Consider $\Sigma=\{a,b,c, d\}$ and $\theta$ be such that $\theta(a)=b$, $\theta(c)=d$. 
\begin{enumerate}
 \item Consider $w = aac$, $u= daa$ and $v=dbb$. Then by Proposition \ref{deftc}, $C_{\theta}(w) = \{ aac, daa, dba, dbb \}$, $C_\theta(u)= \{daa, bda, bbd, bbc\}$ and $C_\theta(v) = \{ dbb, adb, aad, aac\}$. 
 Now, $v \in C_\theta(w)$ and $w \in C_\theta(v)$. On the other hand, $u \in C_\theta(w)$ but $w \notin C_\theta(u)$. Thus, $\theta$-conjugacy is not a symmetric relation.

 \item Consider $w=abb$, $u=bbb$ and $v=bab$. Then by Proposition \ref{deftc}, $C_\theta(w) = \{abb, aab, aaa\}$, $C_\theta(u)=\{ bbb, abb, aab, aaa\}$ and $C_\theta(v) = \{bab, aba, abb\}$. 
 Now, $aab \in C_\theta(w)$, $w \in C_\theta(u)$ and $aab \in C_\theta(u)$. On the other hand, $aaa \in C_\theta(w)$, $w \in C_\theta(v)$ but $aaa \notin C_\theta(v)$. Thus, $\theta$-conjugacy is not a transitive relation.
\end{enumerate}
\end{example}

By definition of conjugates, if $w \in \Sigma^+$ and $a \in \textit{Alph}(w)$, then there exists an element in $C(w)$ that begins with $a$. However, this property does not hold for $\theta$-conjugates, as illustrated in Example \ref{o1uyt}. In the following, we discuss the structure of elements of the set of all $\theta$-conjugates of a word regarding this property.

\begin{lemma}\label{x1}
For $u\in \Sigma^*$ and $a\in \Sigma$,  $C_\theta(ua)=\{ua\}\cup \theta(a) C_\theta(u)$.
\end{lemma}
\begin{proof}
We first show that $C_{\theta}(ua) \subseteq \{ua\}\cup \theta(a)C_{\theta}(u)$. Consider $x \in C_\theta(ua)$. Then for $ua=u'v'$, $x=\theta(v')u'$ where $u', v'\in \Sigma^*$. If $v'=\lambda$, then $x=ua$. Otherwise, if $v'\in \Sigma^+$, then $v'=za$ for some $z\in \Sigma^*$. This implies $x=\theta(a)\theta(z) u'$ and $u=u'z$. Then, $\theta(z)u' \in C_{\theta}(u)$. 
Therefore, $x \in \{ua\}\cup \theta(a) C_\theta(u)$ and  $C_\theta(ua) \subseteq \{ua\}\cup \theta(a) C_\theta(u)$.
Now, consider $y \in \{ua\}\cup \theta(a) C_\theta(u)$. If $y=ua$, then clearly $y \in C_\theta(ua)$. If $y \in \theta(a) C_\theta(u)$, then $y= \theta(a) \theta(q)p$ for $u=pq$, $p, q\in \Sigma^*$. This implies $y \in C_\theta(pqa)$, i.e., $y\in C_\theta(ua)$. Therefore, $\{ua\}\cup \theta(a) C_\theta(u) \subseteq C_\theta(ua)$. As a result, we have, $C_\theta(ua)=\{ua\}\cup \theta(a) C_\theta(u)$. 
\end{proof}

\begin{corollary}\label{corlr121}
 For $v\in \Sigma^+$ and $u\in \Sigma^*$, $ \theta(v) C_\theta(u)$ $\subseteq C_\theta(uv)$.
\end{corollary}
In addition, Proposition \ref{gl1} and Lemma \ref{x1} lead us to the following.
\begin{corollary}\label{jkui87}
For $u\in \Sigma^*$ and $a\in \Sigma$, $C_\theta(ua)= \theta(a) C_\theta(u)$ if and only if $ua$ is a $\theta$-palindrome.
\end{corollary}
We provide some examples to illustrate Lemma \ref{x1} and Corollary \ref{jkui87}.
\begin{example}\label{o1nhgf} 
Consider $\Sigma=\{a,b,c\}$ and $\theta$ be such that $\theta(a)=b$, $\theta(c)=c$. 
\begin{enumerate}
    \item If $w = bccb$, then $C_{\theta}(w) = \{ bccb, abcc, acbc, accb, acca \}$. Note that $bccb$ is not a $\theta$-palindrome. Now, $C_\theta(bcc)=\{bcc, cbc, ccb, cca\}$     and  $C_\theta(bccb)=\{bccb\}\cup \theta(b) C_\theta(bcc)$.
    \item  If $w = abcab$, then $C_{\theta}(w) = \{ abcab, aabca, ababc, abcaa \}$.  Note that $w$ is a $\theta$-palindrome. Now, $C_\theta(abca)=\{abca, babc, bcab, bcaa\}$ and  $C_\theta(abcab)= \theta(b) C_\theta(abca)$.
\end{enumerate}
\end{example}
For $w \in \Sigma^+$, $C(w)$ can contain at most $|w|$ elements and if $|C(w)|=|w|$, then each element of $C(w)$ is primitive.
\cite{thetapalin2021} proved that the set of all $\theta$-conjugates of a word $w$ can have at most $|w|+1$ elements.
We now prove that for $w\in \Sigma^*$, if $|C_\theta(w)|=|w|+1$, then  $C_\theta(w)$ must contain at least one primitive word. To prove this we need the following result.

\begin{lemma}\label{priehtrer1213}
For $a \in \Sigma$ and $w\in \Sigma^n$ with $n\geq 2$, if each element of $C_\theta(w)$ is non-primitive and $a^n \in C_\theta(w)$, then $w=a^n$ with $\theta(a)=a$.
\end{lemma}

\begin{proof}
Consider each element of $C_\theta(w)$ is non-primitive and $a^n \in C_\theta(w)$. Then for some $x, y \in \Sigma^*$, $\theta(y)x= a^n$ and $w=xy$. This implies $\theta(y)=a^j$ and $x= a^i$ for some $i, j \geq 0$ with $i+j=n$. Then, $w=xy=a^i \theta(a)^j$.
 Let us assume that $a \neq \theta(a)$.
If $j=0$, then $w= a^n$ and $\theta(a) a^{n-1} \in C_\theta(w)$ is a primitive word, which is a contradiction. 
Similarly, if $i=0$, then $w= \theta(a)^n$ and $a \theta(a)^{n-1} \in C_\theta(w)$ is a primitive word, which is a contradiction. 
For $i, j \geq 1$,  $w$ is itself a primitive word, which also leads to a contradiction. Therefore, $\theta(a)=a$  and $w=a^n$.
\end{proof}

Using Lemma \ref{priehtrer1213}, we now prove the following. 
\begin{proposition}\label{prinonpp1122}
If $|C_\theta(w)|= n +1$ for a word $w$ of length $n \geq 1$, then there exists a primitive word in $C_\theta(w)$.
\end{proposition}
\begin{proof}
For a word $w$ of length $n \geq 1$, consider  $|C_\theta(w)|= n +1$. 
\begin{itemize}
    \item For $n=1$, the word $w$ is a letter, i.e., a primitive word. Since $w \in C_\theta(w)$, $C_\theta(w)$ contains a primitive word.

    \item 
For  $n \geq 2$, let us assume that each element of $C_\theta(w)$ is non-primitive. 
Now, the length of the root of each element in $C_\theta(w)$ is at most $\lfloor {\frac{|w|}{2}} \rfloor $.
Consider $l=\lfloor {\frac{|w|}{2}} \rfloor $ and $w=w_1 w_2 \cdots w_{l} w_{l+1} \cdots w_n$ where each $w_j \in \Sigma$ for $1 \leq j \leq n$.
For each $i ~(0 \leq i \leq l-1)$, let $\alpha_i = \beta_i \gamma_i$ where $\beta_i= \theta(w_{l-i} \cdots w_n)$ and $\gamma_i =w_1 w_2 \cdots w_{l-i-1} $. 
Then for each $i~ (0 \leq i \leq l-1)$,  $|\beta_i|> l$ and $\alpha_i \in C_\theta(w)$. 
Consider $\rho(\alpha_i)=z_i$ for each $i~ (0 \leq i \leq l-1)$. Then, $|z_i| \leq l$ for each $i~ (0 \leq i \leq l-1)$. 
{Now, each $\alpha_0, \alpha_1, \cdots, \alpha_{l-1}$ has a common prefix $\beta_0$ where $\beta_0 =\theta(w_n) \theta(w_{n-1}) \cdots \theta(w_l)$ and $|\beta_0|>l$. }
This implies each $z_i$  is a proper prefix of $\beta_0$ as $|z_i| \leq l$ and $|\beta_0|>l$ where $0 \leq i \leq l-1$.
Since each element of $C_\theta(w)$ is non-primitive and $|C_\theta(w)|=|w|+1$, 
each $z_i$ ~$(0 \leq i \leq l-1)$ is a distinct prefix of $\beta_0$.
Then, the possible lengths of $z_i$'s are $1, 2, \cdots, l$.

For $i_1 \neq i_2$ $(0 \leq i_1, i_2 \leq l-1)$, if the lengths of $z_{i_1}$ and $z_{i_2}$ 
are equal, then as $z_{i_1}$ and $z_{i_2}$ are prefixes of $\beta_0$, we have $z_{i_1} = z_{i_2}$, i.e.,  $\alpha_{i_1} = \alpha_{i_2}$, which is a contradiction to our assumption that $|C_\theta(w)|=|w|+1$.

Otherwise, consider the length of each $z_i$ is distinct. 
Since the length of the root of each element in $C_\theta(w)$ is at most $l$ and there are $l$ distinct elements $\alpha_0, \alpha_1, \ldots, \alpha_{l-1}$ in $C_\theta(w)$ with a commom prefix $\beta_0$, we have $|\rho(\alpha_{i_3})|=1$ for some $0 \leq i_3 \leq l-1$. Then by Lemma \ref{priehtrer1213}, $w=a^n$ with $\theta(a)=a$, which is a contradiction to our assumption that $|C_\theta(w)|=|w|+1$.
\end{itemize}

Thus, if $|C_\theta(w)|=|w|+1$, then there exists a primitive word in $C_\theta(w)$.
\end{proof}

 For $u, v \in \Sigma^*$, the equality of $C_\theta(u)$ and $C_\theta(v)$ was studied in \cite{thetapalin2021}, and the following result was proved.

 \begin{theorem}\cite{thetapalin2021}\label{math1}
Let $u, v \in \Sigma^*$ such that $C_\theta(u)=C_\theta(v)$. 
\begin{enumerate}
    \item If $|C_\theta(u)|= 1$, then $v=u$.
    \item If $|C_\theta(u)|= 2$, then $v=u$ or $v=\theta(u)$.
    \item If $|C_\theta(u)|\geq 3$, then one of the following holds:
    \begin{enumerate}
        \item $u=v$
        \item $u=a^{n_1} u' a^{n_2}$ and $v=u^R$ where  $a \in \Sigma$, $u' \in \Sigma^+$, $a \notin Pref(u')$, $a \notin Suf(u')$, $u'=u'^R$, $n_1, n_2 \geq 1$, $n_1\neq n_2$, and $\theta(c)=c ~\forall c \in \textit{Alph}(u)$.
    \end{enumerate} 
\end{enumerate}
\end{theorem}

\begin{remark}
   
      \cite{thetapalin2021} proved that if $C_\theta(u)=C_\theta(v)$ and $|C_\theta(u)|\geq 3$ for $u, v \in \Sigma^*$, then either $u=v$ or $u=v^R$ (Proposition 3.15 in \cite{thetapalin2021}). However, the detailed structures of $u$ and $v$ were provided within the proof of Proposition 3.15 in \cite{thetapalin2021}. Since we will later use these structural details, we have explicitly included them in Theorem \ref{math1}. Additionally, to ensure clarity and verifiability, we have added the proof of these structural details in the appendix.

\end{remark}

With the help of a few examples, we now show that the converse of Theorem \ref{math1} is not true.
\begin{example}
Let $\Sigma=\{a, b, c, d\}$ and $\theta$ be such that $\theta(a)=b$, $\theta(c)=c$ and $\theta(d)=d$.
\begin{enumerate}
    \item Consider $u=ca$ and $v=\theta(u)=bc$. Then, $|C_\theta(u)|=2$ and $cb \in C_\theta(v) \setminus C_\theta(u)$.
    \item Consider $u=cdccdc^2$ and $v= u^R = c^2 dccd c$. Then, $|C_\theta(u)|\geq 3$ and  $cdcccdc \in C_\theta(v) \setminus C_\theta(u)$.
\end{enumerate}
Hence, in both cases, $C_\theta(u)\neq C_\theta(v)$.
\end{example}
We now discuss necessary and sufficient conditions for the equality of $C_\theta(u)$ and $C_\theta(v)$ for words $u$ and $v$. 
Since $u=v$ always implies $C_\theta(u)=C_\theta(v)$, we focus on characterizing words $u$ and $v$ such that $u \neq v$ and $C_\theta(u)=C_\theta(v)$.
Furthermore, as $C_\theta(u)=C_\theta(v)$ and $|C_\theta(u)|=1$ implies $u=v$ (by Theorem \ref{math1}, Assertion(1)), we characterize pairs $u, v \in \Sigma^*$ satisfying $u \neq v$, $C_\theta(u)=C_\theta(v)$, and $|C_\theta(u)|=|C_\theta(v)|>1$.
\begin{proposition}\label{cth2}
Consider $u$ and $v$ are two distinct words of same length $k$ such that $|C_\theta(u)|=|C_\theta(v)|=2 $.
Then, $C_\theta(u)=C_\theta(v)$ if and only if $u$ and $v$ satisfy one of the following:
\begin{enumerate}
    \item $u=c$ and $v=\theta(c)$ for some $c \in \Sigma$ with $\theta(c)\neq c$.
     \item $u=ab$ and $v=ba$ for some $a, b \in \Sigma$ with $a \neq b$, $\theta(a)=a$, and $\theta(b)=b $.
\end{enumerate}
\end{proposition}

\begin{proof}
For $u, v \in \Sigma^k$, $u \neq v$  and $|C_\theta(u)|=|C_\theta(v)|=2$, consider $C_\theta(u)=C_\theta(v)$. Then by Theorem \ref{math1},  $v=\theta(u)$ where $\theta(u) \neq u$. 

\begin{itemize}
 \item Consider $|u|=k=1$. Then as $|C_\theta(u)|=2 $, $u=c$ and $v=\theta(c)$ for some $c \in \Sigma$ with $\theta(c)\neq c$.

 \item Consider $|u|=k \geq 2$.  
Since $|C_\theta(u)|=2$ and $u \neq v$, we have the following: for each $\alpha, \beta \in \Sigma^+$ for which $u=\alpha \beta $,   $\theta(\beta)\alpha \in C_\theta(u)$ must be equal to either $u=\alpha \beta$ or $v=\theta(u)=\theta(\beta)\theta(\alpha)$.
 If $ \theta(\beta)\alpha =u$, then by Proposition \ref{gl1}, $u$ is a $\theta$-palindrome, which is a contradiction to the fact that $\theta(u)\neq u$. 
Thus, $ \theta(\beta)\alpha=\theta(u)$, i.e., $\theta(\beta)\alpha = \theta(\beta)\theta(\alpha)$, i.e., $\theta(\alpha)=\alpha$.

Now, for each $\alpha, \beta \in \Sigma^+$ for which $u=\alpha \beta $, we have $\alpha \theta(\beta) \in C_\theta(v) $ (since $v=\theta(u)=\theta(\beta)\theta(\alpha)$).
Since $C_\theta(u)=C_\theta(v)$ and $|C_\theta(u)|=2$,  $\alpha \theta(\beta)$ is equal to either $u$ or $v$. 
 If $\alpha \theta(\beta)=v$, i.e., $\alpha \theta(\beta) = \theta(\beta) \theta(\alpha) $, then  by Proposition \ref{gl1}, $v$ is a $\theta$-palindrome, which is a contradiction to the fact that $\theta(u)\neq u$. 
 Thus, $\alpha \theta(\beta)=u$, i.e., $\alpha \theta(\beta) = \alpha \beta$, i.e., $\theta(\beta)=\beta$.
 
Therefore, for each $\alpha, \beta \in \Sigma^+$ satisfying $u=\alpha \beta$, we have $\theta(\alpha)=\alpha$, $\theta(\beta)=\beta$. 

\begin{itemize}
    \item If $|u|=k=2$, then $u=ab$ and $v=ba$ for some $a, b \in \Sigma$ with $\theta(a)= a$, $\theta(b)=b$.
Here, $a \neq b$, otherwise $u=v=a^2$, which contradicts the fact that $u \neq v$.
    \item If $|u|=k \geq 3$, then as each non-empty proper prefix and suffix of $u$ are $\theta$-palindromes,  $u=v = d^{k}$ for some $d \in \Sigma$ with $\theta(d)=d$, which is a contradiction to the fact that $u \neq v$.
\end{itemize}
\end{itemize}

The converse is straightforward.
\end{proof}

\begin{proposition}\label{cth3}
Consider $u$ and $v$ are two distinct words of same length such that $|C_\theta(u)|=|C_\theta(v)| \geq 3$.
Then, $C_\theta(u)=C_\theta(v)$ if and only if  $u=a^{m+1} b a^{m}$ and $v=a^{m} b a^{m+1}$ where $m \geq 1$, $a, b \in \Sigma$,  $a\neq b$, $\theta(a)=a$ and $\theta(b)=b$.
\end{proposition}
\begin{proof}
For $u, v \in \Sigma^*$, $u \neq v$ and $|C_\theta(u)|=|C_\theta(v)| \geq 3$, consider $C_\theta(u)=C_\theta(v)$.
Then, by Theorem \ref{math1}, $$u=a^{n_1} u' a^{n_2} \text{ and } v=a^{n_2} u' a^{n_1}$$ where $a \in \Sigma$, $n_1, n_2 \geq 1$, $n_1\neq n_2$, $u' \in \Sigma^+$, $a \notin Pref(u')$, $a \notin Suf(u')$,  $u'=u'^R$ and $\theta(c)=c~ \forall c \in \textit{Alph}(u)$.

Without loss of generality, let us take $n_1 > n_2$. 
If $|u'|\geq 2$, 
then according to the above conditions, we have $u'=d_1 \alpha d_1$ for some $d_1 \in \Sigma$ and $\alpha \in \Sigma^*$ with $d_1 \neq a$ and $\theta(d_1)=d_1$.  
This implies $u=a^{n_1} d_1 \alpha d_1 a^{n_2}$ and $v=a^{n_2} d_1 \alpha d_1 a^{n_1}$ where $n_1 > n_2$. Then, $a^{n_1} d_1 \theta(\alpha) a^{n_2} d_1 \in C_\theta(v)\setminus C_\theta(u)$, which is a contradiction to our assumption that $C_\theta(u)=C_\theta(v)$. Thus as $u' \in \Sigma^+$, $|u'|=1$. 
Then, $u=a^{n_1} b a^{n_2}$ and $v=a^{n_2} b a^{n_1}$ where $u'=b \in \Sigma$, $a \neq b$ and $\theta(b)=b$.
If $n_1-n_2 \geq 2$, i.e.,  $n_1 -n_2 = 2+k$ for $k \geq 0$, then $u=a^{n_2+k+2} b a^{n_2}$ and $v=a^{n_2} b a^{n_2+k+2}$. Now, $a^{k+n_2+1} b a^{n_2+1} \in C_\theta(v)\setminus C_\theta(u)$, which is a contradiction to our assumption that $C_\theta(u)=C_\theta(v)$.  Thus, as $n_1>n_2$, $n_1-n_2=1$. Then, $u=a^{n_2+1} b a^{n_2}$ and $v=a^{n_2} b a^{n_2+1}$. 

Conversely, let $u=a^{m+1} b a^{m}$ and $v=a^{m} b a^{m+1}$ where $m \geq 1$, $a, b \in \Sigma$, $\theta(a)=a$, $\theta(b)=b$ and $a\neq b$. 
Then $C_\theta(u)= \{ a^{m+1+j} b a^{m-j}~:~ 0 \leq j \leq m \} \cup \{v\}$,
and $C_\theta(v)=\{v \} \cup \{ a^{m+i} b a^{m+1-i}~:~ 1 \leq i \leq m+1\} \cup \{ a^{m+1} b a^m\}$. 
Since $a^{m+1} b a^m \in \{ a^{m+i} b a^{m+1-i}~:~ 1 \leq i \leq m+1\}$, we have  $C_\theta(v)=\{v \} \cup \{ a^{m+i} b a^{m+1-i}~:~ 1 \leq i \leq m+1\}$. Consider $i=i'+1$. Then, $1 \leq i \leq m+1$ implies $0 \leq i' \leq m$. Thus, $C_\theta(v)=\{v \} \cup \{ a^{m+i'+1} b a^{m-i'}~:~ 0 \leq i' \leq m\}$.
Therefore, $C_\theta(u)=C_\theta(v)$.
\end{proof}

It is well known that for given $u, v \in \Sigma^*$, $C(u)= C(v)$ if and only if $u \in C(v)$. But this statement is not true in general for $\theta$-conjugates of a word.
Propositions \ref{cth2} and  \ref{cth3} describe necessary and sufficient conditions for the equality of $C_\theta(u)$ and $C_\theta(v)$ for distinct words $u$ and $v$ with $|C_\theta(u)|=|C_\theta(v)| \geq 2$. Also,
it is clear from Propositions \ref{cth2} and  \ref{cth3} that if sets of all $\theta$-conjugates of two distinct words  $u \in \Sigma^{\geq 2}$ and  $ v \in \Sigma^{\geq 2}$ are equal, then $u, v$ must be binary words and $\theta (c)=c$ for all $c \in \textit{Alph}(u) \cup \textit{Alph}(v)$.   
Now, combining Propositions \ref{cth2} and \ref{cth3}, we have the following characterizations of $u$ and $v$ such that sets of all $\theta$-conjugates of $u$ and $v$ are equal.

 \begin{theorem}
For $u, v \in \Sigma^*$, $C_\theta(u)=C_\theta(v)$ if and only if  one of the following holds:
\begin{enumerate}
    \item $u=v$.
    \item   $u=ab$ and $v=ba$ for $a, b \in \Sigma$ with $\theta(a)=a$, $\theta(b)=b$ and $a\neq b$.    
    \item $u=c$ and $v=\theta(c)$ for $c \in \Sigma$ with $\theta(c)\neq c$.
    \item    $u=a^{m+1} b a^{m}$ and $v=a^{m} b a^{m+1}$ for $a, b \in \Sigma$ with $\theta(a)=a$, $\theta(b)=b$, $a\neq b$ and $m \geq 1$.
\end{enumerate}
\end{theorem}

\section{$\theta$-conjugates of a Language And Closure Properties of\\ $\theta$-conjugate Operation }\label{sec4}

Conjugate operation or cyclic shift operation (\cite{csre}) is a unary operation on formal languages defined as $C(L) = \bigcup\limits_{w \in L} C(w)$ for some $L \subseteq \Sigma^*$. Here, we call $C(L)$ as the set of all conjugates of the language $L$.
   \cite{csre} showed that the family of regular languages is closed under the conjugate operation. In addition, \cite{ref19} demonstrated that if $L$ is a context-free language, then $C(L)$ is also context-free. {Families of context-sensitive and recursively enumerable languages have also been shown to be closed under the conjugate operation by \cite{cscs}.}
Now, the definition of $\theta$-conjugate of a word  can be naturally extended to a language.
In this section, we first define the $\theta$-conjugate operation on languages and investigate {the closure property of} certain families of languages under the $\theta$-conjugate operation.

{
\begin{definition}\label{q1}
 $\theta$-conjugate operation is a unary operation on formal languages defined as $$C_\theta(L) = \bigcup\limits_{w \in L} C_\theta(w)$$ for some $L \subseteq \Sigma^*$. Here, we call $C_\theta(L)$ as the set of all $\theta$-conjugates of the language $L$.
\end{definition}
}

We now investigate closure properties of certain well-known families of languages with respect to $\theta$-conjugate operation.
For this, we need the following result {whose proof is straightforward}. 
\begin{lemma}\label{mn1}
For a morphic involution $\psi$ and an antimorphic involution $\theta$ over $\Sigma^*$, if $\psi(a)=\theta(a)~ \forall a \in \Sigma$, then $\psi(u)=\theta(u^R)$ for $u\in \Sigma^*$.
\end{lemma}
First, we show that the family of regular languages is closed under $\theta$-conjugate operation. 
 The proof is very similar to the argument presented by \cite{csre} on the closure of regular languages under the conjugate operation.

\begin{theorem}\label{gt1}
If $L \subseteq \Sigma^*$ is regular, then $C_\theta(L)$ is also regular.
\end{theorem}
\begin{proof}
    Consider a deterministic finite automaton (DFA) $A$ which accepts $L$.  Let $Q'$, $q_0$, and $F$ be the set of all states, the initial state, and the set of all final states of $A$, respectively. Now, for each state $q \in Q'$, we first construct two DFAs $C_q$ and $B_q$ with same states and same transition functions as $A$ where $q_0$ is the initial state of $C_q$, $\{q\}$ is the set of all final states of $C_q$, $q$ is the initial state of $B_q$ and $F$ is the set of all final states of $B_q$.
Consider a morphic involution $\psi: \Sigma^*\rightarrow \Sigma^*$ such that $\psi(a)=\theta(a)$ for all $a \in \Sigma$. Then, as the family of regular languages is closed under reversal and morphism \cite{Lothaire1997}, there exists a  DFA $B_q'$ which accepts $\psi((L(B_q))^R)$, i.e., $\theta(L(B_q))$ (by Lemma \ref{mn1}). Now, by definition, a string $w$ is in $C_\theta(L)$ if and only if $w=\theta(v)u$ for $uv\in L$. 
Then, in the accepting computation of $A$ on $uv$, there exists an intermediate state $q$ of $A$ such that the computation of $A$ on $u$ ends in the state $q$, while $v$ is accepted by $A$ starting from $q$.
Then, $\theta(v)$ is accepted by $B_q'$ and $u$ is accepted by $C_q$ and $C_\theta(L)=\bigcup_{q \in Q'}L(B_q')L(C_q)$. Since the family of regular languages is closed under concatenation and union \cite{Lothaire1997}, $C_\theta(L)$ is regular. 
\end{proof}

We now consider the closure property of family of context-free languages under $\theta$-conjugate operation. To proceed, we require the following result, commonly known as the \textit{pumping lemma for Context-Free Languages}:
\begin{theorem}\cite{shallit2008}\label{pumth}
    If $L$ is context-free, then there exists a constant $n$ such that for all $z \in L$ with $|z|\geq n$, there exists a decomposition $z=uvwxy$ with $|vwx| \leq n$ and $|vx| \geq 1$ such that for all $i \geq 0$, we have $u v^i w x^i y \in L$.
\end{theorem}

\begin{theorem}\label{cflwc}
The family of context-free languages is not closed under $\theta$-conjugate operation. 
\end{theorem}

\begin{proof}
Consider the following context-free language $L=\{a^n b^k c^k d^n\; | \; n, k \geq 1\}$ over $\Sigma=\{a, b, c, d\}$.
Suppose $\theta(e)=e$ for all $e \in \Sigma$.
Then,
  $C_\theta(L)=
\{d^{i} a^n b^k c^k d^{n-i}\;|\;  n, k \geq 1, 0\leq i\leq n \}$ $\cup$ 
$\{d^n c^j a^n b^k c^{k-j}\;|\;  n, k \geq 1,  0< j \leq k\}$ $\cup$ 
$\{ d^n c^k b^l a^n b^{k-l}\;|\;  n, k \geq 1,  0<l \leq k\}$.
Consider the regular language
$L'=\{d^{n_1} c^{m_1} a^{m_2} b^{n_2}\;|\;n_1, m_1, \\ m_2, n_2 \geq 1 \}$. Now, if $C_\theta(L)$ is context-free, then as the intersection of a context-free language and a regular language is context-free \cite{shallit2008}, $C_\theta(L)\cap L' = \{  d^n c^k a^n b^k ~|~ n, k \geq 1 \}$ is context-free. But by {pumping lemma} for context-free languages (Theorem \ref{pumth}), $\{  d^n c^k a^n b^k ~|~ n, k \geq 1 \}$ is not context-free (proof is in Appendix), which is a contradiction. Thus, $C_\theta(L)$ is not context-free.
\end{proof}

{We, furthermore, consider the closure property of the family of context-sensitive languages and show that for a given context sensitive language $L_1$, $C_\theta(L_1)$ is also context sensitive by constructing a linear bounded automaton which accepts $C_\theta(L_1)$. In the subsequent discussion, we will establish this result for the alphabet $\Sigma=\{a, b\}$. }

{\begin{theorem}\label{csl091}
If $L_1 \subseteq \{a, b\}^+$ is context-sensitive, then $ C_\theta(L_1)$ is also context-sensitive.
\end{theorem}}
{\begin{proof} 
Let $L_1 \subseteq \{a, b\}^+$ be a context-sensitive language.  Since the family of context-sensitive languages is closed under reversal \cite{rozen1997}, $L_1^R$ is also context-sensitive.
For $\Sigma=\{a, b\}$, let $\psi : \Sigma^* \rightarrow \Sigma^*$ be a morphic involution such that $\psi(c)=\theta(c)$ for all $c\in \Sigma$.
Consider $$L'=\bigcup_{uv \in L_1^R} \{ \psi(u)v^R\; \}.$$
Then, by Lemma \ref{mn1},
$$ L'=\bigcup_{uv \in L_1^R} \{\psi(u)v^R\; \}=  \bigcup_{v^Ru^R \in L_1} \{ \theta(u^R)v^R \;\} = C_\theta(L_1).$$
Now, as $L_1^R $ is a context-sensitive language, there exists a linear bounded automaton 
$M= (Q', \Sigma \cup \{<, >\}, \Gamma, \delta, q'_0, <, >, B, F)$ that accepts $L_1^R$ (\cite{linz2012introduction}) and uses $(k+2)\cdot |z|+2$ tape cells for any $z \in \Sigma^+$ with the initial configuration $q_0'< \underset{(( |z|-1) \text{ no. of blank cells})}{B B \cdots B}  z  \underset{((k\cdot |z|+1) \text{ no. of blank cells})}{B B \cdots B}  >$ where $k$ is a constant. 

Using $M$, we now construct a linear bounded automaton $M'$ which accepts $L'$.
Let $M'= (Q'', \Sigma  \cup \{<, >\}, \Gamma', \delta', q_0,  <, >, B, F)$ where $Q''=Q' \cup \{q_i~|~ 0 \leq i \leq 15\}$, $\Gamma' = \Gamma \cup \{ \#, X \}$, $\delta'(q, c)=\delta(q, c)$ for all $q \in Q'$, $c \in \Gamma$.
For $M'$, we will use $(k+2)\cdot |w| +2$ tape cells for each $w\in \Sigma^*$. Also, for any $w \in \Sigma^*$, we consider the initial configuration of $M'$ as $q_0< \underset{( |w| \text{ no. of blank cells})}{B B \cdots B}  w  \underset{(k\cdot |w| \text{ no. of blank cells})}{B B \cdots B}  >$. 
In the  transition diagram (Figure \ref{LBA1}), we have discussed the transition rules of $M'$ in details. Here we have assumed that $q'_0=q_{15}$. Also, in Figure \ref{LBA1}, the label $a_1|b_1, R$ (resp. $L$ or $N$) on the arc from state  $q_i$ to state $q_j$ signifies that when the machine is in state $q_i$ with the head reading $a_1$, the machine goes to state $q_j$, writes $b_1$, and moves the head to the right (resp. left or stay in its current position).

Suppose $\alpha \in \Sigma^+$ is an input to $M'$.  Then, for every factorization $xy$ of $\alpha$ where $x, y \in \Sigma^*$, the linear bounded automaton $M'$ computes the following:
$q_0< \underset{( |\alpha| \text{ no. of blank cells})}{B B \cdots B}  xy  \underset{(k\cdot |\alpha| \text{ no. of blank cells})}{B B \cdots B}  > {\vdash}^*_{M'}~ q_{15}< \underset{( |\alpha| -1 \text{ no. of blank cells})}{B B \cdots B}  \psi(x)y^R  \underset{(k\cdot |\alpha| +1 \text{ no. of blank cells})}{B B \cdots B}  >$, which is the initial configuration of $M$ for  $\psi(x)y^R$.
 \begin{figure}
    \centering
    \includegraphics[width=.7\textwidth]{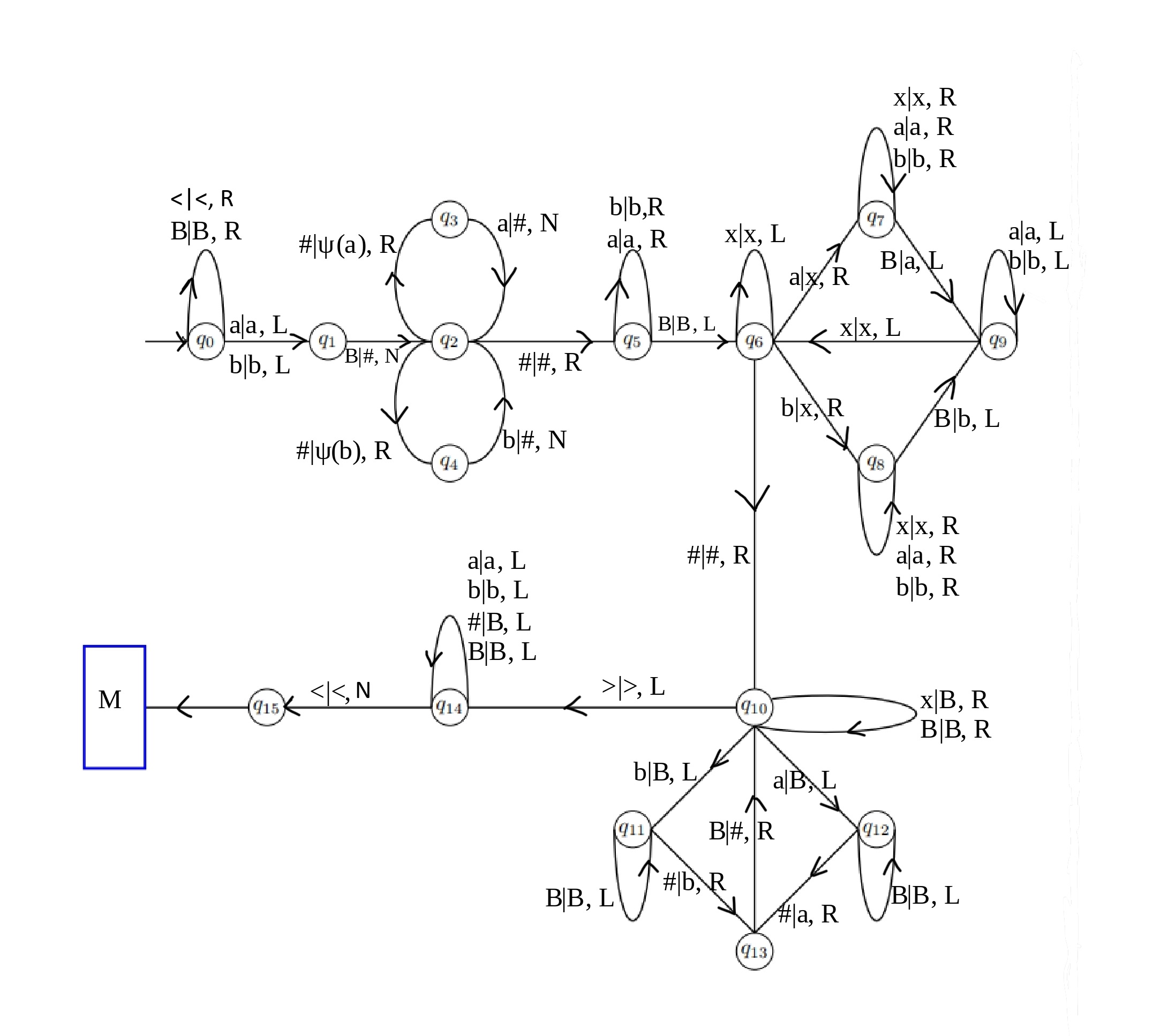}
    \caption{ Transition diagram for $M'$ }
    \label{LBA1}
\end{figure} 
 If, for at least one such factorization $x_1y_1$ of $\alpha$,  $\psi(x_1)y_1^R$ is accepted by $M$, then $M'$ accepts $\alpha$. Conversely, if none of these factorizations lead to acceptance by $M$, then $M'$ rejects $\alpha$. 
 
Now, we show that our machine $M'$ accepts only elements of $L'$. Consider $\beta \in L'$. Then, for some $u_1 v_1 \in L_1^R$, $\beta = \psi(u_1) v_1^R$. If $\beta$ is an input to $M'$, then as $\psi$ is an involution, $M'$ computes the following: \\
$q_0< {B B \cdots B} \psi(u_1) v_1^R {B B \cdots B}  > {\vdash}^*_{M'}~ q_{15}< {B B \cdots B}  u_1v_1  {B B \cdots B}  >$.
Since, $u_1 v_1 \in L_1^R$, $M$ accepts $u_1 v_1$. Thus, $M'$ accepts $\beta$.

Now, consider $\gamma \notin L' $ be an input in $M'$. Then,  for every factorization $\gamma_1 \gamma_2$ of $\gamma$ where $\gamma_1, \gamma_2 \in \Sigma^*$, the linear bounded automaton $M'$ computes the following:\\
$q_0< {B B \cdots B} \gamma_1 \gamma_2 {B B \cdots B}  > {\vdash}^*_{M'}~ q_{15}< {B B \cdots B}  \psi(\gamma_1)\gamma_2^R  {B B \cdots B}  >$.
If, for at least one such factorization $\gamma_1' 
 \gamma_2'$ of $\gamma$, the input $\psi(\gamma_1'){\gamma_2'}^R$ is accepted by $M$, then $\psi(\gamma_1'){\gamma_2'}^R \in L_1^R$. This implies $\gamma_1' \gamma_2' \in L'$, which is a contradiction as $\gamma = \gamma_1' \gamma_2' \notin L' $.
 Thus,  for every factorization $\gamma_1 \gamma_2$ of $\gamma$ where $\gamma_1, \gamma_2 \in \Sigma^*$, the linear bounded automaton $M$ rejects $\psi(\gamma_1)\gamma_2^R$.  Therefore, $M'$ rejects $\gamma$.
\end{proof}

Similar to Theorem \ref{csl091}, we can prove the following result for a context-sensitive language $L$ over an arbitrary alphabet $\Sigma$.

\begin{theorem}\label{csl15643}
If $L \subseteq \Sigma^+$ is context-sensitive language, then $ C_\theta(L)$ is also context-sensitive.
\end{theorem}}

{We now show that if $L$ is a recursively enumerable language, then $C_\theta(L)$ is also recursively enumerable by constructing an unrestricted grammar that generates $C_\theta(L)$. This construction is simple but lengthy. So, we include this construction in the Appendix.
\begin{theorem}\label{csl10089}
 If $L$ is a recursively enumerable language, then $ C_\theta(L)$ is also recursively enumerable. 
\end{theorem}}

\section{Iterated $\theta$-conjugate of Words And Languages}\label{sec5}
In this section, we first define the iterated $\theta$-conjugate of a word and language. Then we characterize the elements of iterated $\theta$-conjugate of a word. We also count the number of elements in the iterated $\theta$-conjugate of a word. Then we study the closure property of iterated $\theta$-conjugate operation.
Throughout this section, we use $|u|_{a, \theta(a)}$  to represent the combined count of letters $a$ and $\theta(a)$ in the string $u \in \Sigma^*$, i.e., 
\begin{equation*}
|u|_{a, \theta(a)}=
    \begin{cases}
        |u|_a+|u|_{\theta(a)} & \text{if } \theta(a) \neq a\\
        |u|_a & \text{if } \theta(a)=a.
    \end{cases}
\end{equation*}

\begin{definition} 
 For a language $L$, we define ${C_\theta}^n(L)$ recursively as ${C_\theta}^{0}(L)= L $ and ${C_\theta}^n(L) = {C_\theta}^{n-1}(C_\theta(L))$. For a word $w$, we define
 ${C_\theta}^n(w)= {C_\theta}^n(\{w\})$, $n \geq 0$.
 
Finally, we define the iterated $\theta$-conjugate of a word $w$ and a language $L$ as ${C_\theta}^*(w)=\bigcup\limits_{n\geq 0} {C_\theta}^n(w)$ and ${C_\theta}^*(L)=\bigcup\limits_{n\geq 0} {C_\theta}^n(L)$, respectively.
\end{definition}

As $w\in C_\theta(w)$ for any $w \in \Sigma^*$, we have the following.
\begin{lemma}\label{yy2}
For $w\in \Sigma^*$ and $m, k\geq 0$, ${C_\theta}^m(w) \subseteq {C_\theta}^{m+k}(w)$.
\end{lemma}

Since the conjugacy relation is an equivalence relation \cite{Lothaire1997}, $u \in C(w)$ implies $C(u)=C(w)$. Then, $C(C(w))=C^2(w)=C(w)$, $C(C^2(w))=C^3(w)=C(w), \ldots, C(C^{i-1}(w))=C^i(w)=C(w)$ for all $i \geq 1$. But this is not true for $\theta$-conjugates of a word. For an example, $ba \in {C_\theta}^2(ab)$ but $ba \notin C_\theta(ab)$ where $\theta(a)=b$. 
We now show that for a given $w \in \Sigma^n$, there exists an integer $m$ depending on $n$ such that 
${C_\theta}^m(w)$ contains  all words $u$ such that $|u|_{a, \theta(a)}$ is equal to $|w|_{a, \theta(a)}$ for all $a \in \Sigma$. To prove this, we need the following results.

From the definition of $\theta$-conjugates, it is clear that if $u$ is an element of ${C_\theta}^m(w)$ for some $m\geq 1$, then for all $a\in \Sigma$, the combined count of letters $a$ and $\theta(a)$ in both $u$ and $w$ must be equal.
\begin{lemma}\label{x11}
If {$u \in {C_\theta}^m(w) $ for any $m \geq 1$,}  then $|u|_{a, \theta(a)} =|w|_{a, \theta(a)}$ for all $a\in \Sigma$.
\end{lemma}
Now, we prove the following.
\begin{lemma}\label{zxz1}
Consider  $u, w\in \Sigma^+$, $z \in \Sigma^*$, $|u|\geq 3$, and $a \in \Sigma$. 
 If $u z \in {C_\theta}^m(w)$ where $m\geq 1$ and  $|u|_{a, \theta(a)}\geq 1$, then $u' z a \in {C_\theta}^{m+4}(w)$ for some $u'\in \Sigma^+$.
\end{lemma}
\begin{proof}
Let   $u\in \Sigma^+$ and $u z \in {C_\theta}^m(w)$ where $|u|_{a, \theta(a)}\geq 1$. We now have the following cases. 
\begin{itemize}
    \item If $|u|_{a}\geq 1$, then for some $\alpha, \beta \in \Sigma^*$, $u=\alpha a \beta$ and
\begin{align*}
    \alpha a \beta z \in {C_\theta}^m(w)
   & \implies \theta(z) \theta(\beta) \alpha a \in {C_\theta}^{m+1}(w)\\
    &\implies \theta(a) \theta(z) \theta(\beta) \alpha \in {C_\theta}^{m+2}(w) \\
    &\implies \theta(\alpha) \beta z a \in {C_\theta}^{m+3}(w)\\
    &\implies \theta(\alpha) \beta z a \in {C_\theta}^{m+4}(w) & (\text{by Lemma } \ref{yy2})
\end{align*}
    \item If $|u|_{\theta(a)}\geq 1$,  then for some $\alpha', \beta' \in \Sigma^*$, $u=\alpha' \theta(a) \beta'$ and
\begin{align*}
    \alpha' \theta(a) \beta' z \in {C_\theta}^m(w)
   & \implies \theta(z) \theta(\beta') a \alpha'  \in {C_\theta}^{m+1}(w)\\
    &\implies \theta(\alpha') \theta(a) \theta(z) \theta(\beta') \in {C_\theta}^{m+2}(w) \\
    &\implies \beta' z a \alpha' \in {C_\theta}^{m+3}(w)\\
    &\implies \theta(\alpha') \beta' z a \in {C_\theta}^{m+4}(w).
\end{align*}
\end{itemize}
Hence, {we have the desired result}.
\end{proof}

For any $v \in \Sigma^2$ and $a, b \in \Sigma$, if $|v|_{a, \theta(a)}\geq 1$ and $|v|_{b,\theta(b)}\geq 1$, then the following can be shown using a similar proof technique of Lemma \ref{zxz1}.

\begin{lemma}\label{ityr5r45r}
Consider $v, z, w \in \Sigma^+$, $|v|=2$ and $a, b \in \Sigma$.
If $vz \in {C_\theta}^m(w)$ where $m\geq 1$, $|v|_{a, \theta(a)}\geq 1$ and $|v|_{b,\theta(b)}\geq 1$, then $ z ab \in {C_\theta}^{m+4}(w)$.
\end{lemma}
With the help of Lemmas \ref{x11}, \ref{zxz1}, and \ref{ityr5r45r},  we now prove that if $w$ is a word of length at least three, then for all integer $m \geq 4|w|-6$,  ${C_\theta}^m(w)$ is equal to the collection of all strings $u$ such that  $|u|_{a, \theta(a)}$ is equal to $|w|_{a, \theta(a)}$ for all $ a \in \Sigma$.

\begin{proposition}\label{xt1}
If $w\in \Sigma^n$ with $n \geq 3$, then for all $m \geq 4n-6$,
${C_\theta}^{m}(w)=\{u\;|\; |u|_{a, \theta(a)}=|w|_{a, \theta(a)} \; \forall a \in \Sigma \}$.
\end{proposition}

\begin{proof}
Let $w$ be a word of length $n \geq 3$ and $T =\{u\;|\; |u|_{a, \theta(a)}=|w|_{a, \theta(a)} ~\forall a \in \Sigma \}$. {By Lemma \ref{x11}, we know that ${C_\theta}^{m_1}(w) \subseteq T$ for any $m_1 \geq 1$. Then,  ${C_\theta}^{m}(w) \subseteq T$ for all $m \geq 4n-6$.}
Now, we prove that $T \subseteq {C_\theta}^{m}(w)$ for all $m \geq 4n-6$. 
Consider $x \in T$. 
Then, as $ |x|_{a, \theta(a)}=|w|_{a, \theta(a)} \;\text{ for all } a \in \Sigma$, we have, $|x|=|w|=n\geq 3$.
Let $x=x_1 x_2 \cdots x_n$ where each $x_j \in \Sigma$ for $1\leq j \leq n$. 
Since $ |x|_{x_i, \theta(x_i)}=|w|_{x_i, \theta(x_i)} \;\text{ for each } i (1\leq i \leq n) $, it follows that $w$ must contain at least one occurrence of  $x_i$ or $\theta(x_i)$ for each  $i (1\leq i \leq n)$.
Then, for each $i$ $(1\leq i \leq n)$, $w$ can be written as $w= \alpha_i x_i \beta_i$ or $w=\alpha_i \theta(x_i) \beta_i$ for some $\alpha_i, \beta_i \in \Sigma^*$.

If $w=\alpha_1 x_1 \beta_1$, then $\theta(\beta_1) \alpha_1 x_1 \in C_\theta(w)$. Since for each $i$ $(2 \leq i \leq n-2)$, either $|\theta(\beta_1) \alpha_1|_{x_i}\geq 1$ or $|\theta(\beta_1) \alpha_1|_{\theta(x_i)}\geq 1$, by Lemma \ref{zxz1}, we have $ v x_1 x_2 \cdots x_{n-2} \in {C_\theta}^{1+4(n-3)}(w)$, i.e., $ v x_1 x_2 \cdots x_{n-2} \in {C_\theta}^{4n-11}(w)$ for some $v \in \Sigma^+$ with $|v|=2$. Now, $|v|_{x_n, \theta(x_n)}\geq 1$ and $|v|_{x_{n-1}, \theta(x_{n-1})}\geq 1$. Then, by Lemma \ref{ityr5r45r}, $x_1 x_2 \cdots x_{n-2} x_{n-1} x_n \in {C_\theta}^{4n-11+4}(w)$, i.e.,  $x_1 x_2 \cdots x_{n-2} x_{n-1} x_n \in {C_\theta}^{4n-7}(w)$.
Therefore, by Lemma \ref{yy2}, $x_1 x_2 \cdots x_{n-2} x_{n-1} x_n \in {C_\theta}^{4n-6}(w)$.

If $w=\alpha_1 \theta(x_1) \beta_1$, then $\theta(\beta_1) x_1 \alpha_1 \in C_\theta(w)$ and $\theta(\alpha_1) \theta(\beta_1) x_1 \in {C_\theta}^2(w)$. Then, similar to the above,  $x_1 x_2 \cdots x_{n-2} x_{n-1} x_n \in {C_\theta}^{2+4(n-3)+4}(w)$, i.e.,  $x_1 x_2 \cdots x_{n-2} x_{n-1} x_n \in {C_\theta}^{4n-6}(w)$. 

Therefore, $x \in {C_\theta}^{4n-6}(w)$.
Then, by Lemma \ref{yy2}, $ x \in {C_\theta}^{m}(w)$ for all $m \geq 4n-6$.
Thus, $T \subseteq {C_\theta}^{m}(w)$ for all $m \geq 4n-6$. Hence, {we have the desired result}.
\end{proof}

One can observe by direct computation that if $w  \in \Sigma$ and $w' \in \Sigma^2$, then ${C_\theta}(w)={C_\theta}^*(w)$ and ${C_\theta}^3(w')={C_\theta}^*(w')$.
For words $u$ of length at least three, we have the following result from  Proposition \ref{xt1}.
\begin{theorem}\label{vx1}
If $w\in \Sigma^n$ with $n \geq 3$, then ${C_\theta}^{4n-6}(w) = {C_\theta}^{*}(w)= \{u\;|\; |u|_{a, \theta(a)}=|w|_{a, \theta(a)} \; \forall a \in \Sigma \}$.
\end{theorem}
\begin{proof}
From the definition of iterated $\theta$-conjugate of $w$, it is clear that ${C_\theta}^{4n-6}(w) \subseteq {C_\theta}^{*}(w)$. Consider $x \in {C_\theta}^{*}(w)$. Then, for some integer $k$, $x \in {C_\theta}^{k}(w)$. 
If $k < 4n-6$, then by Lemma \ref{yy2}, ${C_\theta}^{k}(w) \subseteq {C_\theta}^{4n-6}(w)$. This implies $x \in  {C_\theta}^{4n-6}(w)$.
If $k \geq 4n-6$, then by Proposition \ref{xt1}, ${C_\theta}^{4n-6}(w) = {C_\theta}^{k}(w)= \{u\;|\; |u|_{a, \theta(a)}=|w|_{a, \theta(a)} \; \forall a \in \Sigma \}$. This implies $x \in  {C_\theta}^{4n-6}(w)$.
Thus, ${C_\theta}^{*}(w) \subseteq {C_\theta}^{4n-6}(w)$. Hence, ${C_\theta}^{4n-6}(w) = {C_\theta}^{*}(w)$.
\end{proof}

We now provide an example to show that for a word $w$ of length $n \geq 3$, the bound $4n-6$ is optimal.  Additionally, we present an example to show that the smallest $i$ such that  ${C_\theta}^{i}(w) = {C_\theta}^{*}(w) = \{u\;|\; |u|_{a, \theta(a)}=|w|_{a, \theta(a)} \; \forall a \in \Sigma \}$ may or may not be equal to $4n-6$.

\begin{example}\label{bnvmr50}
Let $w \in \Sigma^*$ and $\theta$ be such that $\theta(a)=b$, $\theta(c)=d$ and $\theta(g)=h$. 
\begin{enumerate}
    \item Consider $w=cag$. Then,\\
$C_\theta(w)=\{ hca, cag, hbd, hbc  \},$\\ 
${C_\theta}^2(w)= C_\theta(w) \cup \{  bdg, bdh, cah, dhb, dah, bhc, chb, dag  \}, $\\
$ {C_\theta}^3(w)= {C_\theta}^2(w) \cup  \{ hcb, gcb, gbc, gda, dbh, dgb, gca, dga, hda, ach, agc, gbd, \\agd, adh \},$\\
${C_\theta}^4(w) = {C_\theta}^3(w) \cup  \{ bgd, ahc,  cgb, gdb, gad, bgc, adg, bch, cha, bhd, dha, bcg,\\ gac, ahd \}$,\\
${C_\theta}^5(w) = {C_\theta}^4(w) \cup \{  cbh, cga, acg, had, cbg, dbg, hdb \}$,\\
${C_\theta}^6(w) = {C_\theta}^5(w)  \cup  \{ hac \}$, and
${C_\theta}^i(w)={C_\theta}^6(w) $ for all $i \geq 6$.\\
Hence, ${C_\theta}^*(w) = {C_\theta}^6(w)$. Note that, for $m=6 = 4|w|-6 $, ${C_\theta}^*(w) = {C_\theta}^m(w)$.

        \item   Consider $u= caa$. Then, $C_\theta(w)=\{ caa, bca, bbc, bbd  \},$\\ 
        ${C_\theta}^2(w)= C_\theta(w) \cup \{  bdb, cbb, cab,  daa, dab, dbb, bda  \} $, \\
        $ {C_\theta}^3(w)= {C_\theta}^2(w) \cup  \{ bcb, abc, adb, acb, aad,  ada, aca, abd, aac \}$,\\
        ${C_\theta}^4(w) = {C_\theta}^3(w) \cup  \{  bac, dba, bad, cba \}$, and
        ${C_\theta}^i(w)={C_\theta}^4(w) $ for all $i \geq 4$. \\
      Hence, ${C_\theta}^*(w) = {C_\theta}^4(w)$. Note that, for $m=4 < 4|u|-6 = 6$, ${C_\theta}^*(w) = {C_\theta}^m(w)$.
\end{enumerate}
\end{example}

For $w\in \Sigma^*$, ${C_\theta}^*(w)=\{u\;|\; |u|_{a, \theta(a)}=|w|_{a, \theta(a)}~ \forall  a \in \Sigma \}$.
Using basic counting principles, we now compute the number of elements in ${C_\theta}^*(w)$ for $w \in \Sigma^n$.

\begin{theorem}\label{coit3421}
For $w \in \Sigma^n$, let $\textit{Alph}(w) = A_1 \cup A_2 \cup B$ where $A_1 \cap A_2 = \emptyset$, $A_1 \cap B = \emptyset$, $A_2 \cap B =\emptyset$, $B = \{ b  : \theta(b) = b \} $, and $A_2 = \{\theta(a) : a \in A_1 \text{ and } \theta(a)\in \textit{Alph}(w) \}$.
 Consider $A_1 = \{a_{1}, a_{2}, \cdots, a_{l} \}$ and $B = \{b_{1}, b_{2}, \cdots, b_{r} \}$ for some non-negative integers $l$ and $r$.
For each $a \in A_1$ and $b \in B$, let $|w|_{a, \theta(a)} = n_a$ and $|w|_{b, \theta(b)} = m_b$, respectively. Consider $\sum\limits_{a\in A_1} n_a = n'$. Then 
$$|{C_\theta}^*(w)|= {n \choose n_{a_{1}}\; n_{a_{2}} \; \cdots \; n_{a_{l}} \; m_{b_{1}}\; m_{b_{2}} \; \cdots \; m_{b_{r}}} \times 2^{n'}.$$
\end{theorem}

\begin{proof}
    Each word $u \in {C_\theta}^*(w)$ satisfies the condition $|u|_{c, \theta(c)} = |w|_{c, \theta(c)}$ for all $c \in \Sigma$.
    This means that what truly matters for counting the words in ${C_\theta}^*(w)$ is the total number of occurrences of each letter and its counterpart under $\theta$.

 Consider $\textit{Alph}(w) = A_1 \cup A_2 \cup B$ where $A_1 \cap A_2 = \emptyset$, $A_1 \cap B = \emptyset$, $A_2 \cap B =\emptyset$, $B = \{ b  : \theta(b) = b \} = \{b_{1}, b_{2}, \cdots, b_{r} \}$, $A_1 = \{a_{1}, a_{2}, \cdots, a_{l} \}$ and $A_2 = \{\theta(a) : a \in A_1 \text{ and } \theta(a)\in \textit{Alph}(w) \}$.
    Then $A_1$ is the collection of all those elements of $\textit{Alph}(w)$ such that for any two distinct indices $i$ and $j$ with $1 \leq i, j \leq l$, we have $a_i \neq \theta(a_i)$, $a_j \neq \theta(a_j)$ and $a_i \neq \theta(a_j)$.\\
    Let $C_1=\{ c_{a_1}, c_{a_2}, \cdots, c_{a_l}\}$ and $D=\{ c_{b_1}, c_{b_2}, \cdots, c_{b_r} \}$ such that $|C_1|= |A_1|=l$, $|D|=|B|=r$, $ C_1 \cap D = \emptyset$,  $\Sigma \cap C_1 = \emptyset$ and $\Sigma \cap D = \emptyset$.
    To count the words of ${C_\theta}^*(w)$ systematically, we follow the following steps:
    
    \textbf{Step 1: }
    
    For each $a_i \in A_1$ ($1 \leq i \leq l$),  we first replace each $a_i$ and $\theta(a_i)$ in $w$ by $C_{a_i}$.
    For each $b_j \in B$ ($1 \leq j \leq r$),  we also replace each $b_j$ in $w$ by $C_{b_j}$.
    Then $w$ becomes a word over the alphabet $C_1 \cup D$. Now by the given conditions,   $|w|_{C_{a_i}} = n_{a_i}$ and $|w|_{C_{b_j}} = m_{b_j}$ for each $1 \leq i \leq l$ and $1 \leq j \leq r$.
    \vspace{.4cm}  \\ 
    \fbox{\parbox{\textwidth}{ Explanation of Step 1 through an Example: Consider $w= a_1 b_1 \theta(a_1) \theta(a_2)$ where $a_1, b_1, a_2 \in \Sigma$, $a_1 \neq \theta(a_1)$, $a_2 \neq \theta(a_2)$, and $b_1 = \theta(b_1)$. Then, $A_1 = \{a_1, \theta(a_2) \}$, $A_2 = \{ \theta(a_1) \} $ and $B=\{ b_1\}$.
    Let  $C_1=\{ c_{a_1}, c_{\theta(a_2)}\}$ and $D=\{ c_{b_1} \}$.  
    Now, $a_1$ and $\theta(a_1)$ is replaced by $C_{a_1}$; $\theta(a_2)$ is replaced by $C_{\theta(a_2)}$; and $b_1$ is replaced by $C_{b_1}$. Then, $w$ becomes $ C_{a_1} C_{b_1} C_{a_1} C_{\theta(a_2)}$.   }}
    \\

    \textbf{Step 2: }

    Now, we collect all the possible arrangements of $w$ in a set $S$. Using multinomial coefficient, we have 
    $$|S| =  {n \choose n_{a_{1}}\; n_{a_{2}} \; \cdots \; n_{a_l} \; m_{b_{1}}\; m_{b_{2}} \; \cdots \; m_{b_{r}}}.$$
    \vspace{.4cm}\\    
    \fbox{\parbox{\textwidth}{Explanation of Step 2 through an Example: After Step 1, we have $ w = C_{a_1} C_{b_1} C_{a_1} C_{\theta(a_2)}$. Then, $S=\{ C_{a_1} C_{b_1} C_{a_1} C_{\theta(a_2)}, 
C_{a_1} C_{b_1} C_{\theta(a_2)} C_{a_1},  C_{a_1} C_{a_1} C_{b_1} C_{\theta(a_2)}, 
C_{a_1} C_{a_1} C_{\theta(a_2)} C_{b_1}, 
C_{a_1} C_{\theta(a_2)} C_{a_1} C_{b_1}, \\
C_{a_1} C_{\theta(a_2)} C_{b_1} C_{a_1}, 
C_{b_1} C_{a_1} C_{a_1} C_{\theta(a_2)}, 
C_{b_1} C_{a_1} C_{\theta(a_2)} C_{a_1}, 
C_{b_1} C_{\theta(a_2)} C_{a_1} C_{a_1}, 
C_{\theta(a_2)} C_{a_1} C_{a_1} C_{b_1}, \\
C_{\theta(a_2)} C_{a_1} C_{b_1} C_{a_1}, 
C_{\theta(a_2)} C_{b_1} C_{a_1} C_{a_1} \}$. Using multinomial coefficients, $|S|={4 \choose 2~ 1~ 1 } =12$.  }}
    \\

\textbf{Step 3: }

Consider $\alpha \in S$. Then  $|\alpha|_{C_{a_i}} = n_{a_i}$ and $|\alpha|_{C_{b_j}} = m_{b_j}$ for each $1 \leq i \leq l$ and $1 \leq j \leq r$.
 We now replace each $C_{a_i}$ in $\alpha$ by one of $a_i$ and $\theta(a_i)$ where $1 \leq i \leq l$.
Thus, each position of $\alpha$ with $C_{a_i}$ has two possible replacements $a_i$ and $\theta(a_i)$ where $1 \leq i \leq l$.
Similarly, we replace each $C_{b_j}$ in $\alpha$ by $b_j$ where $1 \leq j \leq r$.
Thus, each position of $\alpha$ with $C_{b_j}$ has exactly one possible replacement $b_j$ where $1 \leq j \leq r$.
Then as $\sum\limits_{i= 1}^{l} n_{a_i} = n'$, we get $2^{n'}$ distinct words over $\Sigma$ from $\alpha$. Let's keep these $2^{n'}$ distinct words in a new set $R$. For each $\alpha \in S$, we get $2^{n'}$ distinct elements, and we put them in the set $R$. Then by basic counting principle,  $$ |R| =  {n \choose n_{a_{1}}\; n_{a_{2}} \; \cdots \; n_{a_l} \; m_{b_{1}}\; m_{b_{2}} \; \cdots \; m_{b_{r}}} \times 2^{n'}.$$ 
    \fbox{\parbox{\textwidth}{ Explanation of Step 3 through an Example: Consider $\alpha = C_{a_1} C_{a_1} C_{b_1} C_{\theta(a_2)}$. Then, $C_{a_1}$ is replaced by one of $a_1$ and $\theta(a_1)$.
    $C_{\theta(a_2)}$ is replaced by one of $\theta(a_2)$ and $a_2$. $C_{b_1}$ is replaced by $b_1$.
    Since $\sum\limits_{a \in A_1} n_{a} = 3$,
    from $\alpha$ we got the following  $2^3$ elements: $a_1 a_1 b_1 \theta(a_2)$, $a_1 a_1 b_1 a_2$, 
  $a_1 \theta(a_1) b_1 a_2$, $a_1 \theta(a_1) b_1 \theta(a_2)$,
$\theta(a_1) a_1 b_1 a_2$, $\theta(a_1) a_1 b_1 \theta(a_2)$, $\theta(a_1) \theta(a_1) b_1 a_2$,
$\theta(a_1) \theta(a_1) b_1 \theta(a_2)$. So, $|R|=12 \times 8=96$.
    }}
    \\

\textbf{Step 4: }

We now show that $R= C_\theta^*(w)$. Consider $u \in R$. Then by the above construction, $|u|_{a, \theta(a)} = |w|_{a, \theta(a)}$ for all $a \in \Sigma$. Thus, $u \in C_\theta^*(w)$. Therefore,  $R \subseteq C_\theta^*(w)$.
Consider $v \in  C_\theta^*(w)$. Then, $|v|_{a, \theta(a)} = |w|_{a, \theta(a)}$ for all $a \in \Sigma$. Now, we  replace each $a_i$ and $\theta(a_i)$ in $v$ and $w$ by $C_{a_i}$ where $1 \leq i \leq l$. Also, we replace each $b_j$ in $v$ and $w$ by $C_{b_j}$ where $1 \leq j \leq r$.  After replacements, consider $v$ and $w$ become $v'$ and $w'$, respectively, over $C_1 \cup D$. Since
 $|v|_{a, \theta(a)} = |w|_{a, \theta(a)}$ for all $a \in \Sigma$, $v'$ is a permutation of $w'$, i.e., $v' \in S$. Then by step 3, $v \in R$. Thus, $ C_\theta^*(w) \subseteq R$. Therefore, $R= C_\theta^*(w)$.    
\end{proof}

We demonstrate Theorem \ref{coit3421} using the following example.

\begin{example}\label{ex1} 
Let $w \in \Sigma^*$ and $\theta$ be such that $\theta(a)=b$, $\theta(c)=d$, $\theta(g)=h$ and $\theta(e)=e$.
\begin{enumerate}
    \item  Consider $w=cag$. Then $A_1=\{a, c, g\}$, $A_2= \emptyset$, $B=\emptyset$,
     $|w|_{a, b}=n_a=1$, $|w|_{c, d}=n_c=1$, $|w|_{g, h}=n_g=1$, and $n_a+ n_c +n_g =n' =3$. 
     Thus by Theorem \ref{coit3421}, $|{C_\theta}^*(w)|= {3 \choose 1\; 1 \; 1} \times 2^{3}$ $ = 48$. 
    In  Example \ref{bnvmr50}, we have explicitly computed the set ${C_\theta}^*(cag)$. From this set, we can verify that $|{C_\theta}^*(cag)|=48$.

    \item Consider $w= caa$. Then, $A_1=\{a, c\}$, $A_2= \emptyset$, $B=\emptyset$,
    $|w|_{c, d}=n_c=1$, $|w|_{a, b}=n_a=2$, and $n_c+ n_a = n' = 3$. 
     Thus by Theorem \ref{coit3421}, $|{C_\theta}^*(w)|= {3 \choose 1\; 2 } \times 2^{3}$ $ = 24$. 
    In  Example \ref{bnvmr50}, we have explicitly computed the set ${C_\theta}^*(caa)$.  From this set, we can verify that $|{C_\theta}^*(caa)|=24$.
    
    \item Consider $w=aee$. Then, $A_1=\{a\}$, $A_2= \emptyset$, $B=\{e\}$,
    $|w|_{a, b}=n_a=1$, $|w|_{e}=m_e=2$, and $n_a = n' = 1$. 
    Thus by Theorem \ref{coit3421}, $|{C_\theta}^*(w)|= {3 \choose 1\; 2 } \times 2^{1}$ $ = 6$. 
    We now verify the cardinality of ${C_\theta}^*(w)$ by computing it.
    Here, $C_\theta(w) = \{  eeb, eea, eae, aee \}$, ${C_\theta}^2(w)= \{ eeb, eea, bee, eae, ebe, aee  \}$, and ${C_\theta}^i(w)={C_\theta}^2(w)$ for all $i \geq 2$. Thus, ${C_\theta}^*(w)={C_\theta}^2(w)$ and $|{C_\theta}^*(w)|=6$.
\end{enumerate}

\end{example}

Consider $L$ is a finite language and $ \max\{|w| : w \in L\} \leq 2$.  Then by direct computation,
${C_\theta}^{3}(L) = {C_\theta}^*(L) = \bigcup\limits_{w \in L} \{u ~|~ |u|_{a, \theta(a)}=|w|_{a, \theta(a)} ~ \forall a \in \Sigma\}$. If the length of the largest word in $L$ is greater than two, then by Theorem \ref{vx1}, we have the following remark. 
\begin{remark}\label{lsas}
If $L$ is a finite language and $n = \max\{|w| : w \in L\} \geq 3$, then ${C_\theta}^{4n-6}(L) = {C_\theta}^*(L) = \bigcup\limits_{w \in L} \{u ~|~ |u|_{a, \theta(a)}=|w|_{a, \theta(a)} ~ \forall a \in \Sigma \}$.
\end{remark}
It can be observed from Remark \ref{lsas} that if $L$ is finite, then ${C_\theta}^{i}(L)$ is finite for  $i \geq 0$. Also,  if ${C_\theta}^{i}(L)$ is finite for some  $i \geq 0$, then as $L \subseteq {C_\theta}^{i}(L)$, $L$ is finite. This implies that $L$ is finite if and only if ${C_\theta}^{i}(L)$ is finite for $i \geq 0$.

We now study closure properties of families of regular and context-free languages under iterated $\theta$-conjugate operation.
\begin{lemma}
Regular languages and context-free languages are not closed under iterated $\theta$-conjugate operation. Furthermore, the iterated $\theta$-conjugate of a regular language is not necessarily context-free.
\end{lemma}
\begin{proof}
Consider a regular language $L = \{(bce)^*\}$ over $\{ b, c, d, e, f\}$ where $\theta(b)=b$, $\theta(c)=d$ and $\theta(e)=f$. Since ${C_\theta}^{*}(w)=\{u\;|\; |u|_{a, \theta(a)}=|w|_{ a, \theta(a)} ~\forall  a \in \Sigma \}$, ${C_\theta}^{*}(L)=\bigcup\limits_{w \in L} \{u\;|\;|u|_{a, \theta(a)}=|w|_{a, \theta(a)} ~ \forall  a \in \Sigma  \}.$
Then, ${C_\theta}^{*}(L)  \cap \{ b^{m_1} c^{m_2} e^{m_3}\;|\; m_i\geq 0\} = \{ b^n c^n e^n \;|\; n\geq 0 \}$, which is not context-free (using Theorem \ref{pumth}). Note that $\{ b^{m_1} c^{m_2} e^{m_3}\;|\; m_i\geq 0\}$ is a regular language.
Since the regular languages and the context-free languages are closed under intersection with regular languages (\cite{shallit2008}), the claim holds.
\end{proof}

\section{$\theta$-conjugate-free Languages}\label{sec6}
In this section, we first define a $\theta$-conjugate-free language $L$, i.e., there are no two distinct strings in $L$ such that one is a $\theta$-conjugate of the other. 
Then for given $u, v \in \Sigma^n$, we discuss a linear time algorithm to decide whether or not $v$ belongs to  $C_\theta(u)$.
Finally, we decide the existence of $u, v$ where $ u \neq v$ and $v \in C_\theta(u)$  in different language classes, i.e., we investigate some decidability problems for $\theta$-conjugate-freeness for different language classes.
\begin{definition}
A language $L$ is \textit{$\theta$-conjugate-free} if for all $w\in L$, $C_\theta(w) \cap L = \{w\}$.
\end{definition}

\subsection{Decidability of $\theta$-conjugate-freeness}

We first examine a straightforward scenario where we have two strings of equal length. Our task is to decide whether one string is a $\theta$-conjugate of the other. Specifically, given two strings $u$ and $v$, we need to decide whether or not $v$ belongs to  $C_\theta(u)$. 
A similar question for conjugates can be decided in linear time, i.e., given two strings $u$ and $v$ of equal length, we can decide whether $v \in C(u)$  in linear time. This can be done by checking whether $v$ is a substring of $uu$.

For a word $w \in \Sigma^*$, let $w[i]$ and $w[i\cdots j]$ denote the letter of the string $w$ at position $i$ and the factor $w[i]w[i+1]\cdots w[j]$ of $w$, respectively, where $1 \leq i \leq j \leq |w|$.
 To check whether $v$ belongs to  $C_\theta(u)$, we first check the equality of $u$ and $v$. If $u=v$, then $v \in C_\theta(u)$. Otherwise, we compute $\theta(u)$ and compare two strings, $v$ and $\theta(u)$. Now,  we find the first index where $v$ and $\theta(u)$ have different characters.  Let this index be $i_L$. If such an index does not exist, then  $v=\theta(u)$ and $v \in C_\theta(u)$. If $i_L=1$, then $v \notin C_\theta(u)$. 
 If $i_L >1$, then we check the equality of substrings $\theta(v[i_L \cdots |v|])$ and $\theta(u)[i_L \cdots |u|]$. If they are equal, then $v=\theta(y)x$ and $\theta(u)=\theta(y) \theta(x)$, i.e., $v \in C_\theta(u)$. Otherwise, $v \notin C_\theta(u)$. See the Algorithm \ref{algo1} for details. 
From the analysis of Algorithm \ref{algo1}, we conclude that determining whether $v$ is in $C_\theta(u)$ can be done in linear time.

\begin{algorithm}
\caption{A linear time algorithm for deciding $v \in C_\theta(u)$}\label{algo1}
\begin{algorithmic}
\Require Two strings of $u$ and $v$ of same length $n$ and mapping $\theta$
 \State Define an empty array $A$ of size $n$
 \State $z \gets {\lambda}$  
 \State $s \gets 1$
\State $k \gets 0$
\State $k_1 \gets 0$
\For{\texttt{$j=1$ to $n$}}
    \If{$u[j]=v[j]$} $k \gets k+1$
    \EndIf 
\EndFor
\If{$k=n$} return true \Comment{Then, $v=u$}
\Else
\For{\texttt{$l=n$ to $1$}}
    $z \gets z ~||~ \theta(u[l])$ \Comment{ $||$ indicates concatenation operation}
\EndFor
\EndIf
\State $\theta(u) \gets z$
\If{$v[1] \neq \theta(u)[1]$} return false
     \Else
     \For{\texttt{$i=2$ to $n$}}
           \If{$\theta(u)[i] \neq v[i]$}
                $A[s] \gets i$;
                $s \gets s+1$
                \For{\texttt{$l=i$ to $n$}}
                    \If{$\theta(u)[i] \neq \theta(v[n+A[1]-i])$} return false
                    \EndIf
                \EndFor\\
                \hspace{1.6cm}return true \Comment{$\theta(u)=\theta(y)\theta(x)$, $v= \theta(y) x$ for some $x, y \in \Sigma^+$}
                \Else\\
              \hspace{1.5cm}  {$k_1 \gets k_1+1$}
            \EndIf   
      \EndFor
      \If{$k_1=n-1$} return true \Comment{Then, $v=\theta(u)$}
      \EndIf
\EndIf
\end{algorithmic}
\end{algorithm}

\begin{theorem}
If we have two strings $u$ and $v$ of equal length $n$, then we can decide whether or not $v$ is a member of $C_\theta(u)$ in $O(n)$ time.
\end{theorem}
We now investigate decidability problems for $\theta$-conjugate-freeness of $L$ when $L$ is either regular or context-free.
 \cite{decreg} proved that it is decidable to determine whether or not a regular language contains two distinct conjugates of a word. We now prove that it is also decidable to determine whether or not a regular language $L$ with no $\theta$-palindrome contains two distinct $\theta$-conjugates of a word.
\begin{theorem}
If $L$ is a regular language with no $\theta$-palindrome,  then it is decidable to determine whether or not $L$ is $\theta$-conjugate-free.
\end{theorem}
\begin{proof}
Consider $L$ is a regular language  with no $\theta$-palindromes.
From Proposition \ref{x1}, we know that $C_\theta(u'a)=\{u'a\}\cup \theta(a) C_\theta(u')$ for $u'\in \Sigma^*$ and $a\in \Sigma$. 
Then, $C_\theta(L)=\bigcup\limits_{w \in L }C_\theta(w)  = L \cup \bigcup\limits_{\substack{u'a \in L \\ u' \in \Sigma^*, a \in \Sigma}}\theta(a)C_\theta(u') = L \cup T$ where  $T=\bigcup\limits_{\substack{u'a \in L \\ u' \in \Sigma^*, a \in \Sigma}}\theta(a)C_\theta(u')$.
Now, using a similar proof technique of Theorem \ref{gt1}, it is not hard to prove that if $L$ is regular, then $T$ is also regular. 
We now prove that $L$ is $\theta$-conjugate-free if and only if $L\cap T = \emptyset$. 
Consider $L$ is  $\theta$-conjugate-free. 
Then for all $w\in L$, $C_\theta(w) \cap L = \{w\}$. 
If $x \in L\cap T$, then $x \in L$ and $x \in T$. 
This implies $x \in \theta(a) C_\theta(z)$ for some  $za \in L$.  
If $x=za$, then for $\alpha \in \Sigma^+$ and $\beta \in \Sigma^*$, $x=\theta(\alpha)\beta$, $za=\beta \alpha$ and $\theta(\alpha)\beta =\beta \alpha$ which implies $x$ is a $\theta$-palindrome, a contradiction to the fact that $L$ does not contain any $\theta$-palindrome. Thus, $x \neq za$, i.e., $ \{x, za\}  \subseteq C_\theta(za)\cap L $ which implies $L$ is not  $\theta$-conjugate-free, a contradiction. Hence, $L\cap T = \emptyset$. Now, consider $L\cap T = \emptyset$. If $w, w' \in L$ such that $w\neq w'$ and $w' \in C_\theta(w)$, then $w'\in T$ which implies $w' \in L\cap T$, a contradiction. Thus, $L$ is $\theta$-conjugate-free. Therefore,  $L$ is $\theta$-conjugate-free if and only if $L \cap T = \emptyset$. 
Now, $L$ and $T$ are regular languages. 
Since the emptiness of the intersection of two regular languages is decidable (\cite{shallit2008}), we can decide whether or not L is $\theta$-conjugate-free.
\end{proof}

Now, we discuss the decidability result for a context-free language.
First, we recall the following. An instance of the Post's Correspondence Problem (PCP) \cite{post1946} consists of a positive integer $n$, and two ordered $n$-tuples
of strings $( U, V )$, where $U = (u_0, u_1, \ldots, u_{n-1})$ and $V = (v_0, v_1, \ldots, v_{n-1})$, $u_i, v_i \in \Sigma^*$, $0\leq i \leq n-1$. A solution for the PCP instance $(U, V )$ is a sequence of integers $i_1, i_2, \ldots, i_k$, $0\leq i_j \leq n-1$, $j=1, 2, \ldots, k$, $k\geq 1$, such that
$$u_{i_1} u_{i_2} \cdots u_{i_k} = v_{i_1} v_{i_2} \cdots v_{i_k}.$$
\begin{proposition}\cite{post1946}
The decision problem of determining whether or not a given PCP instance has a solution is undecidable.
\end{proposition}

Using PCP, we can show that it is undecidable to determine whether or not a context-free language contains two distinct conjugates of a word (the proof is provided in the Appendix). 
We now prove that determining the $\theta$-conjugate-freeness of a given context-free language is also undecidable by reducing the PCP to this problem.

\begin{theorem}
It is undecidable to determine {whether or not} a given context-free
language $L$ is $\theta$-conjugate-free.
\end{theorem}

\begin{proof}
Let $\Sigma$ be alphabet and $(U, V)$ be an instance of PCP problem where $U=(u_0, u_1, \ldots, u_{n-1})$, $V=(v_0, v_1, \ldots, v_{n-1})$ with each $u_i, v_i \in \Sigma^*$. 
Assume that the symbols $\#, \$, a_0, a_1, \cdots, a_{n-1}$ are distinct and not in $\Sigma$. Consider $\Sigma'=\Sigma \cup \{ \#, \$, a_0, a_1, \cdots, a_{n-1}\}$.
Let $\theta : \Sigma'^* \rightarrow \Sigma'^*$ be an antimorphic involution defined as 
\begin{enumerate}
    \item  $\theta(a) \in \Sigma$  for all $a \in \Sigma$
    \item $\theta(\$)=\#$
    \item $\theta(c)=c$ for all $c \in \{ a_0, a_1, \cdots, a_{n-1}\}$.
\end{enumerate}

We define a context-free grammar $G = (N, \Sigma', P, S)$ where $N = \{S, N_U, N_V\}$ is the set of non-terminals, $\Sigma'$ is the set of terminals, $S$ is the start symbol and $P$ is the set of derivation rules which contains the following rules:
\begin{itemize}
    \item $ S \rightarrow \# N_U \# ~ | ~ \$  N_V \$ $
    \item $N_U \rightarrow u_0 N_U a_0 ~ |~u_1 N_U a_1 ~ |~ \cdots~|~  u_{n-1} N_U a_{n-1} ~ | ~ \lambda$
    \item $N_V \rightarrow a_0 N_V \theta(v_0) ~|~ a_1 N_V \theta(v_1) ~|~ \cdots ~|~ a_{n-1} N_V \theta(v_{n-1}) ~|~ \lambda $
\end{itemize}

Now, $L(G)$ denotes the set of all words generated by the grammar $G$ and any string of $L(G)$ is in one of the following forms :
\begin{equation}
   \# u_{j_{1}} u_{j_{2}} \cdots u_{j_{l}}         a_{j_l} \cdots  a_{j_2} a_{j_1} \# 
   \tag{$A$}
\end{equation}
or
\begin{equation*}
    \$  a_{j_1}  a_{j_2} \cdots  a_{j_l} \theta(v_{j_l}) \cdots \theta(v_{j_2}) \theta(v_{j_1})   \$ 
    \tag{$B$}
\end{equation*}

Under our assumed conditions on $\#, \$, a_0, a_1, \cdots, a_{n-1}, \theta, u_0, u_1, \cdots u_{n-1}, v_0, v_1, \cdots, v_{n-1}$, we observe that 
the symbol $\#$ appears only at the left and right ends of words in the form $A$,
the symbol $\$$ does not appears in words of the form ${A}$, 
the symbol $\$$ appears only at the left and right ends of words in the form ${B}$, 
and the symbol $\#$ does not appears in words of the form ${B}$.
We now show that $L(G)$ is not $\theta$-conjugate-free if and only if the PCP instance $(U, V)$ has a solution.

First, we prove that if the PCP instance $(U, V)$ has a solution, then $L(G)$ is not $\theta$-conjugate-free. Assume that the PCP instance $(U, V)$ has a solution. Then there is a sequence of integers $ i_0, i_1, \ldots, i_k$ such that
$u_{i_0} u_{i_1} \cdots u_{i_k} = v_{i_0} v_{i_1} \cdots v_{i_k}$.
Now, $L(G)$ contains the following elements:
\begin{align*}
       \beta' =   \$  a_{i_0}  a_{i_1} \cdots  a_{i_k} \theta(v_{i_k}) \cdots \theta(v_{i_1}) \theta(v_{i_0}) \$  
   &=  \$  a_{i_0}  a_{i_1} \cdots  a_{i_k}  \theta( v_{i_0} v_{i_{1}} \cdots   v_{i_{k}}) \$ \\
   &= \$  a_{i_0}  a_{i_1} \cdots  a_{i_k} \theta( u_{i_0} u_{i_{1}} \cdots   u_{i_{k}})\$ 
   \end{align*}
  \text{ and }
   \begin{align*}
    \alpha' = \#  u_{i_0} u_{i_{1}} \cdots u_{i_{k}}  a_{i_k} \cdots a_{i_1} a_{i_0} \#. 
   \end{align*}
As $\theta(\$)=\#$ and $\theta(c)=c ~ \forall c \in \{ a_0, a_1, \cdots, a_{n-1}\}$, we have $\beta'=\theta(\alpha')$.
Since $\# \neq \$ $, $\alpha' \neq \beta'$. 
Thus, $L(G)$ is not $\theta$-conjugate-free.

Now, we prove the other part, i.e., if $L(G)$ is not $\theta$-conjugate-free, then the PCP instance $(U, V)$ has a solution. Assume that $L(G)$ is not $\theta$-conjugate-free, i.e., there exist $w_1, w_2$ in $L(G)$ such that $w_2 \in C_\theta(w_1)$ and $w _1 \neq w_2$. Then $w_1=xy$ and $w_2=\theta(y) x$  for some $x \in {\Sigma'}^*$ and $ y \in {\Sigma'}^+$.
Since elements of $L(G)$ are in the form $A$ or $B$, we have two cases:
\begin{itemize}
    \item \textbf{Case I :} Consider $w_1=xy$ is in the form $A$. 
    Then for some integer $m$, 
    $$w_1=xy= \# u_{i_{1}} u_{i_{2}} \cdots u_{i_{m}}   a_{i_m} \cdots  a_{i_2} a_{i_1} \# .$$ 
    As $\theta(\#)= \$ $ and $y \in {\Sigma'}^+$, $\theta(y)x$ starts with $\$ $. 
    This implies $\theta(y)x$ must be in the form $B$. 
    Then for some integer $r$,
    $$ w_2=\theta(y)x = \$  a_{j_1}  a_{j_2} \cdots  a_{j_r} \theta(v_{j_r}) \cdots \theta(v_{j_2}) \theta(v_{j_1})   \$. $$
    Since $ |w_1|_{\$} =0$, $\theta(\#)= \$ $ and $w_2 \in C_\theta(w_1)$,  we have $w_2=\theta(w_1)$. Then using the values of $w_2$ and $w_1$, we have 
    \begin{align*}
        & \$  a_{j_1}  a_{j_2} \cdots  a_{j_r} \theta(v_{j_r}) \cdots \theta(v_{j_2}) \theta(v_{j_1})   \$ = \theta(\# u_{i_{1}} u_{i_{2}} \cdots u_{i_{m}}   a_{i_m} \cdots  a_{i_2} a_{i_1} \# )\\
     \implies & \$  a_{j_1}  a_{j_2} \cdots  a_{j_r} \theta(v_{j_r}) \cdots \theta(v_{j_2}) \theta(v_{j_1})   \$ = \$  a_{i_1}  a_{i_2} \cdots  a_{i_m}  \theta(u_{i_m}) \cdots \theta(u_{i_2}) \theta(u_{i_1})   \$.
    \end{align*}
    Since $\{a_0, a_1, \cdots, a_{n-1}\} \cap \Sigma = \emptyset$, 
    each $u_i, v_i \in \Sigma^*$, 
    and $\theta(a) \in \Sigma ~\forall a \in \Sigma$,  we have,     $r=m$, ${i_1}={j_1}$, ${i_2}={j_2}$, \ldots, ${i_m}={j_m}$ and $v_{i_1} v_{i_2} \cdots v_{i_m} = u_{i_1} u_{i_2} \cdots u_{i_m}$. Hence, the PCP instance has a solution.
    
    \item \textbf{Case II :} Consider $w_1=xy$ is in the form $B$. 
    Then similar to Case I, we show that the PCP instance has a solution.  
\end{itemize}
Therefore, as PCP is undecidable \cite{post1946}, determining whether or not $L$ is $\theta$-conjugate-free is also undecidable when $L$ is context-free.
\end{proof}

\section{Conclusions}
This work explores various properties of the set of all $\theta$-conjugates of a word and language.
 First, we have discussed the structure of $\theta$-conjugates of a word and the solution of $C_\theta(u)=C_\theta(v)$. 
 Then we have investigated closure properties of certain families of languages under $\theta$-conjugate operation. 
 Next, we have studied the iterated $\theta$-conjugate of words and languages.
Finally, we have discussed the concept of $\theta$-conjugate-free language and some related decidability problems.  
For regular language $L$ with no $\theta$-palindrome, it is decidable to determine whether or not $L$ is $\theta$-conjugate-free. But for an arbitrary regular language $L$, the answer of the following question is still unknown.
\vspace{.2cm}\\
\textbf{Question 1:}  For a regular language $L$, is it decidable to determine whether or not $L$ is $\theta$-conjugate-free? 
\vspace{.2cm}\\

\acknowledgements
\label{sec:ack}
We sincerely thank the anonymous referee for his/her careful reading and valuable comments on the revised paper, which helped to improve the paper in detail.

\nocite{*}
\bibliographystyle{abbrvnat}
\bibliography{sample-dmtcs}

\begin{thebibliography}{40}
\providecommand{\natexlab}[1]{#1}
\providecommand{\url}[1]{\texttt{#1}}
\expandafter\ifx\csname urlstyle\endcsname\relax
  \providecommand{\doi}[1]{doi: #1}\else
  \providecommand{\doi}{doi: \begingroup \urlstyle{rm}\Url}\fi

\bibitem[Adleman(1994)]{adleman1994molecular}
L.~M. Adleman.
\newblock Molecular computation of solutions to combinatorial problems.
\newblock \emph{Science}, 266\penalty0 (5187):\penalty0 1021--1024, 1994.

\bibitem[Aoi et~al.(1998)Aoi, Yoshinobu, Tanizawa, Kinoshita, and Iwasaki]{aoi1998solution}
Y.~Aoi, T.~Yoshinobu, K.~Tanizawa, K.~Kinoshita, and H.~Iwasaki.
\newblock Solution of the knapsack problem by deoxyribonucleic acid computing.
\newblock \emph{Japanese Journal of Applied Physics}, 37\penalty0 (10R):\penalty0 5839--5841, 1998.

\bibitem[Belkheyar et~al.(2023)Belkheyar, Daemen, Dobraunig, Ghosh, and Rasoolzadeh]{belkheyar2023introducing}
Y.~Belkheyar, J.~Daemen, C.~Dobraunig, S.~Ghosh, and S.~Rasoolzadeh.
\newblock Introducing two low-latency cipher families: Sonic and supersonic.
\newblock \emph{Cryptology ePrint Archive}, 2023.

\bibitem[Bi et~al.(2016)Bi, Yue, Wu, and Ye]{bi2016initiator}
S.~Bi, S.~Yue, Q.~Wu, and J.~Ye.
\newblock Initiator-catalyzed self-assembly of duplex-looped {DNA} hairpin motif based on strand displacement reaction for logic operations and amplified biosensing.
\newblock \emph{Biosensors and Bioelectronics}, 83:\penalty0 281--286, 2016.

\bibitem[Braich et~al.(2002)Braich, Chelyapov, Johnson, Rothemund, and Adleman]{braich2002solution}
R.~S. Braich, N.~Chelyapov, C.~Johnson, P.~W. Rothemund, and L.~Adleman.
\newblock Solution of a 20-variable 3-{SAT} problem on a {DNA} computer.
\newblock \emph{Science}, 296\penalty0 (5567):\penalty0 499--502, 2002.

\bibitem[Brandst\"{a}dt(1981)]{cscs}
A.~Brandst\"{a}dt.
\newblock Closure properties of certain families of formal languages with respect to a generalization of cyclic closure.
\newblock \emph{RAIRO. Informatique Th{\'e}orique}, 15\penalty0 (3):\penalty0 233--252, 1981.

\bibitem[Chiniforooshan et~al.(2010)Chiniforooshan, Kari, and Xu]{Ehsan}
E.~Chiniforooshan, L.~Kari, and Z.~Xu.
\newblock Pseudo-power avoidance.
\newblock In \emph{Developments in Language Theory}, volume 6224 of \emph{Lecture Notes in Computer Science}, pages 432--433. Springer, 2010.

\bibitem[Czeizler et~al.(2010)Czeizler, Kari, and Seki]{CZ12}
E.~Czeizler, L.~Kari, and S.~Seki.
\newblock On a special class of primitive words.
\newblock \emph{Theoretical Computer Science}, 411\penalty0 (3):\penalty0 617--630, 2010.

\bibitem[Endrullis et~al.(2017)Endrullis, Shallit, and Smith]{decreg}
J.~Endrullis, J.~Shallit, and T.~Smith.
\newblock Undecidability and finite automata.
\newblock In \emph{Developments in Language Theory}, volume 10396 of \emph{Lecture Notes in Computer Science}, pages 160--172. Springer, 2017.

\bibitem[Gawrychowski et~al.(2014)Gawrychowski, Manea, and Nowotka]{gawrychowski2014testing}
P.~Gawrychowski, F.~Manea, and D.~Nowotka.
\newblock Testing generalised freeness of words.
\newblock In \emph{31st International Symposium on Theoretical Aspects of Computer Science (STACS 2014)}, volume~25 of \emph{Leibniz International Proceedings in Informatics (LIPIcs)}, pages 337--349, 2014.

\bibitem[Head et~al.(2000)Head, Rozenberg, Bladergroen, Breek, Lommerse, and Spaink]{head2000computing}
T.~Head, G.~Rozenberg, R.~S. Bladergroen, C.~Breek, P.~Lommerse, and H.~P. Spaink.
\newblock Computing with {DNA} by operating on plasmids.
\newblock \emph{Biosystems}, 57\penalty0 (2):\penalty0 87--93, 2000.

\bibitem[Henkel et~al.(2005)Henkel, Bladergroen, Balog, Deelder, Head, Rozenberg, and Spaink]{henkel2005protein}
C.~V. Henkel, R.~S. Bladergroen, C.~I. Balog, A.~M. Deelder, T.~Head, G.~Rozenberg, and H.~P. Spaink.
\newblock Protein output for {DNA} computing.
\newblock \emph{Natural Computing}, 4:\penalty0 1--10, 2005.

\bibitem[Henkel et~al.(2007)Henkel, B{\"a}ck, Kok, Rozenberg, and Spaink]{henkel2007dna}
C.~V. Henkel, T.~B{\"a}ck, J.~N. Kok, G.~Rozenberg, and H.~P. Spaink.
\newblock {DNA} computing of solutions to knapsack problems.
\newblock \emph{Biosystems}, 88\penalty0 (1-2):\penalty0 156--162, 2007.

\bibitem[Jeddi and Saiz(2017)]{jeddi2017three}
I.~Jeddi and L.~Saiz.
\newblock Three-dimensional modeling of single stranded {DNA} hairpins for aptamer-based biosensors.
\newblock \emph{Scientific Reports}, 7\penalty0 (1):\penalty0 1178, 2017.

\bibitem[Jir\'{a}skov\'{a} and Okhotin(2008)]{csre}
G.~Jir\'{a}skov\'{a} and A.~Okhotin.
\newblock State complexity of cyclic shift.
\newblock \emph{RAIRO-Theoretical Informatics and Applications}, 42\penalty0 (2):\penalty0 335--360, 2008.

\bibitem[Johnson(2008)]{johnson2008automating}
C.~R. Johnson.
\newblock Automating the {DNA} computer: solving n-variable 3-{SAT} problems.
\newblock \emph{Natural Computing}, 7:\penalty0 239--253, 2008.

\bibitem[Kari and Kulkarni(2014{\natexlab{a}})]{kari2014generating}
L.~Kari and M.~S. Kulkarni.
\newblock Generating the pseudo-powers of a word.
\newblock \emph{Journal of Automata, Languages and Combinatorics}, 19\penalty0 (1-4):\penalty0 157--171, 2014{\natexlab{a}}.

\bibitem[Kari and Kulkarni(2014{\natexlab{b}})]{kari2014pseudo}
L.~Kari and M.~S. Kulkarni.
\newblock Pseudo-identities and bordered words.
\newblock In \emph{Discrete Mathematics and Computer Science}, pages 207--222. Editura Academiei Române, Romania, 2014{\natexlab{b}}.

\bibitem[Kari and Kulkarni(2017)]{kari2017disjunctivity}
L.~Kari and M.~S. Kulkarni.
\newblock Disjunctivity and other properties of sets of pseudo-bordered words.
\newblock \emph{Acta Informatica}, 54\penalty0 (4):\penalty0 379--398, 2017.

\bibitem[Kari and Mahalingam(2007)]{kari2007involutively}
L.~Kari and K.~Mahalingam.
\newblock Involutively bordered words.
\newblock \emph{International Journal of Foundations of Computer Science}, 18\penalty0 (05):\penalty0 1089--1106, 2007.

\bibitem[Kari and Mahalingam(2008{\natexlab{a}})]{kari2008watson}
L.~Kari and K.~Mahalingam.
\newblock Watson-crick bordered words and their syntactic monoid.
\newblock \emph{International Journal of Foundations of Computer Science}, 19\penalty0 (05):\penalty0 1163--1179, 2008{\natexlab{a}}.

\bibitem[Kari and Mahalingam(2008{\natexlab{b}})]{watson}
L.~Kari and K.~Mahalingam.
\newblock Watson-{C}rick conjugate and commutative words.
\newblock In \emph{{DNA} Computing}, volume 4848 of \emph{Lecture Notes in Computer Science}, pages 273--283. Springer, 2008{\natexlab{b}}.

\bibitem[Kari and Mahalingam(2010)]{LKari2010}
L.~Kari and K.~Mahalingam.
\newblock Watson-{C}rick palindromes in {DNA} computing.
\newblock \emph{Natural Computing}, 9\penalty0 (2):\penalty0 297--316, 2010.

\bibitem[Kari and Ng(2019)]{kari2019state}
L.~Kari and T.~Ng.
\newblock State complexity of pseudocatenation.
\newblock In \emph{Language and Automata Theory and Applications}, volume 11417 of \emph{Lecture Notes in Computer Science}, pages 203--214. Springer, 2019.

\bibitem[Kari and Seki(2009)]{kari2009pseudoknot}
L.~Kari and S.~Seki.
\newblock On pseudoknot-bordered words and their properties.
\newblock \emph{Journal of Computer and System Sciences}, 75\penalty0 (2):\penalty0 113--121, 2009.

\bibitem[Linz(2012)]{linz2012introduction}
P.~Linz.
\newblock \emph{An introduction to formal languages and automata (5th edition)}.
\newblock Jones \& Bartlett Learning, 2012.

\bibitem[Lothaire(1997)]{Lothaire1997}
M.~Lothaire.
\newblock \emph{Combinatorics on words}.
\newblock Cambridge University Press, 1997.

\bibitem[Lyndon and Sch\"{u}tzenberger(1962)]{Schutz62}
R.~C. Lyndon and M.~P. Sch\"{u}tzenberger.
\newblock The equation {$a^{M}=b^{N}c^{P}$} in a free group.
\newblock \emph{Michigan Mathematical Journal}, 9:\penalty0 289--298, 1962.

\bibitem[Mahalingam et~al.(2022)Mahalingam, Maity, and Pandoh]{thetapalin2021}
K.~Mahalingam, A.~Maity, and P.~Pandoh.
\newblock Counting ({W}atson-{C}rick) palindromes in {W}atson-{C}rick conjugates.
\newblock \emph{Information and Computation}, 285:\penalty0 104863, 2022.

\bibitem[Oshiba(1972)]{ref19}
T.~Oshiba.
\newblock Closure property of the family of context-free languages under the cyclic shift operation.
\newblock \emph{Electronics \& Communications in Japan}, 55\penalty0 (4):\penalty0 119--122, 1972.

\bibitem[Ouyang et~al.(1997)Ouyang, Kaplan, Liu, and Libchaber]{ouyang1997dna}
Q.~Ouyang, P.~D. Kaplan, S.~Liu, and A.~Libchaber.
\newblock {DNA} solution of the maximal clique problem.
\newblock \emph{Science}, 278\penalty0 (5337):\penalty0 446--449, 1997.

\bibitem[Post(1946)]{post1946}
E.~L. Post.
\newblock A variant of a recursively unsolvable problem.
\newblock \emph{Bulletin of the American Mathematical Society}, 52\penalty0 (4):\penalty0 264--268, 1946.

\bibitem[Rozenberg and Salomaa~(Eds.)(1997)]{rozen1997}
G.~Rozenberg and A.~Salomaa~(Eds.).
\newblock \emph{Handbook of Formal Languages}, volume~1.
\newblock Springer, 1997.

\bibitem[Sakamoto et~al.(2000)Sakamoto, Gouzu, Komiya, Kiga, Yokoyama, Yokomori, and Hagiya]{sakamoto2000molecular}
K.~Sakamoto, H.~Gouzu, K.~Komiya, D.~Kiga, S.~Yokoyama, T.~Yokomori, and M.~Hagiya.
\newblock Molecular computation by {DNA} hairpin formation.
\newblock \emph{Science}, 288\penalty0 (5469):\penalty0 1223--1226, 2000.

\bibitem[Shallit(2008)]{shallit2008}
J.~Shallit.
\newblock \emph{A second course in formal languages and automata theory}.
\newblock Cambridge University Press, 2008.

\bibitem[Shi et~al.(2017)Shi, Zhou, Huang, Zhang, and Li]{shi2017recent}
B.~Shi, Y.~Zhou, Y.~Huang, J.~Zhang, and X.~Li.
\newblock Recent advances on the encoding and selection methods of {DNA}-encoded chemical library.
\newblock \emph{Bioorganic \& Medicinal Chemistry Letters}, 27\penalty0 (3):\penalty0 361--369, 2017.

\bibitem[Starosta(2011)]{starosta2011theta}
{\v{S}}.~Starosta.
\newblock On theta-palindromic richness.
\newblock \emph{Theoretical Computer Science}, 412\penalty0 (12-14):\penalty0 1111--1121, 2011.

\bibitem[Walter et~al.(2015)Walter, Wright, and Nobel]{walter2015consistent}
V.~Walter, F.~A. Wright, and A.~B. Nobel.
\newblock Consistent testing for recurrent genomic aberrations.
\newblock \emph{Biometrika}, 102\penalty0 (4):\penalty0 783--796, 2015.

\bibitem[Wang et~al.(2008)Wang, Bao, Hu, Wang, and Zhan]{wang2008solving}
X.~Wang, Z.~Bao, J.~Hu, S.~Wang, and A.~Zhan.
\newblock Solving the {SAT} problem using a {DNA} computing algorithm based on ligase chain reaction.
\newblock \emph{Biosystems}, 91\penalty0 (1):\penalty0 117--125, 2008.

\bibitem[Zhang and Seelig(2011)]{zhang2011dynamic}
D.~Y. Zhang and G.~Seelig.
\newblock Dynamic {DNA} nanotechnology using strand-displacement reactions.
\newblock \emph{Nature Chemistry}, 3\penalty0 (2):\penalty0 103--113, 2011.

\end{thebibliography}
\label{sec:biblio}


\section{Appendix}

\textbf{A0.} For $u, v \in \Sigma^*$ with $u \neq v$, consider $C_\theta(u)=C_\theta(v)$ and $|C_\theta(u)|\geq 3$. Then
 $u=a^{n_1} u' a^{n_2}$ and $v=u^R$ where $a \in \Sigma$, $u' \in \Sigma^+$, $a \notin Pref(u')$, $a \notin Suf(u')$, $u'=u'^R$, $n_1, n_2 \geq 1$, $n_1\neq n_2$, and $\theta(c)=c$ for all $c \in \textit{Alph}(u)$.

\begin{proof}
    Let $u = u_1u_2\cdots u_n$ and $v = v_1v_2 \cdots v_n$ be words of length $n$ such that $u_i, v_i \in \Sigma$. 
    
    If $\theta(u_n)\neq u_1$, then in $C_\theta(u)$, $u$ is the only element having $u_1$ as a prefix and the rest of the elements have $\theta(u_n)$ as a prefix.
    If $\theta(v_n)= v_1$, then all elements of $C_\theta(v)$ have the same letter as a prefix, which is a contradiction to the fact that $C_\theta(u) = C_\theta(v)$. This implies, $\theta(v_n) \neq v_1$. Then, $v$ is the only element in $C_\theta(v)$ that starts with $v_1$. Also, $u$ is the only element in $C_\theta(u)$ that starts with $u_1$. As $|C_\theta(u)|\geq 3$ and $C_\theta(u) = C_\theta(v)$, we get, $u=v$, which is a contradiction.
    Thus, $\theta(u_n) = u_1$. Then all elements in $C_\theta(u)$ start with $u_1$. Since, $C_\theta(u) = C_\theta(v)$, we get, $u_1=v_1$ and $\theta(u_n)=u_1=\theta(v_n)$, i.e., $u_n=v_n$. Let 
    \begin{equation}\label{eqqe1}
        u = u_1^{n_1} u' u_n^{n_2} \text{ and } v = u_1^{m_1} v' u_n^{m_2}
    \end{equation}
    where $u', v' \in \Sigma^*$, $u_1 \notin Pref(u')$,  $u_1 \notin Pref(v')$,  $u_n \notin Suf(u')$,  $u_n \notin Suf(v')$, and $n_1, n_2, m_1, m_2 \geq 1$.\\
    Without loss of generality, let $n_1+n_2\geq m_1+m_2$. Since $ u = u_1^{n_1} u' u_n^{n_2}$, for an element in $C_\theta(u)$ to start with $u_1^{n_1+n_2}$, it must be of the form $\theta(y u_n^{n_2})u_1^{n_1}x$ where $xy=u'$.
    From Eq. (\ref{eqqe1}), it is clear that $u_n$ is not present as a suffix of $y$ and $u_1$ is not present as a prefix of $x$. This implies that $u_1^{n_1+n_2}u'$ is the only element in $C_\theta(u)$ having $u_1^{n_1+n_2}$ as a prefix. Similarly, $u_1^{m_1+m_2}v'$ is the only element in $C_\theta(v)$ having $u_1^{m_1+m_2}$ as a prefix. If $n_1+n_2>m_1+m_2$, then $u_1^{n_1+n_2}u' \notin C_\theta(v)$, which is a contradiction to the fact that
    $C_\theta(u) = C_\theta(v)$. This implies $n_1+n_2=m_1+m_2$. Also, we get,  $u_1^{n_1+n_2}u'= u_1^{m_1+m_2}v'$, i.e., $u'=v'$. Thus, we get,
    \begin{equation*}
        u=u_1^{n_1} u' u_n^{n_2} \text{ and }  v=u_1^{m_1} u' u_n^{m_2}
    \end{equation*}
    and $n_1+n_2=m_1+m_2$.  If $n_1=m_1$, then $n_2=m_2$. This implies $u$=$v$, which is a contradiction. 
    Thus, $n_1\neq m_1$. 
    Without loss of generality,  let $n_1>m_1$ and $m_2>n_2$. We have the following cases:
     \begin{itemize}
        \item If $u'=\lambda$, then $u=u_1^{n_1} u_n^{n_2}$ and $v=u_1^{m_1}  u_n^{m_2}$. If $u_1=u_n$, then $u=v$, a contradiction.  If $u_1 \neq u_n$, then as $n_1>m_1\geq 1$, we have, $n_1\geq 2$ and $u_1^{n_2} u_n u_1^{n_1-1} \in C_\theta(u)\setminus C_\theta(v)$, which is a contradiction.
     \item If $u'\in \Sigma^+$, then $u=u_1^{n_1} u' u_n^{n_2}$ and $v=u_1^{m_1} u' u_n^{m_2}$. We have the following cases:
      \begin{itemize}
            \item $n_1\neq m_2$: If $u_1 \neq u_n$, then as $n_1>m_1\geq 1$, we have, $n_1\geq 2$ and $u_1^{n_2} \theta(u') u_n u_1^{n_1-1} \in C_\theta(u)\setminus C_\theta(v)$, which is a contradiction. Otherwise, $u_1=u_n$. Then $u=u_1^{n_1} u' u_1^{n_2}$ and $v=u_1^{m_1} u' u_1^{m_2}$. If $n_1>m_2$, then $m_1>n_2$ and $u_1^{n_2} \theta(u') u_1^{n_1} \in C_\theta(u)\setminus C_\theta(v)$, which is a contradiction. If $n_1<m_2$, then $n_2>m_1$, $u=u_1^{n_1} u' u_1^{m_1+k}$ and  $v=u_1^{m_1} u' u_1^{n_1+k}$ for $k\geq 1$.
            We have the following cases:
            \begin{itemize}
                \item $\theta(u')\neq u' :$ Then, $u_1^{n_1+k} \theta(u') u_1^{m_1} \in C_\theta(v)\setminus C_\theta(u)$, which is a contradiction.
                \item $\theta(u')=u' :$ Then, as $n_1>m_1$, $v\in C_\theta(v)\setminus C_\theta(u) $, which is a contradiction.
            \end{itemize}    
        \item $n_1=m_2$:  Then, $n_2=m_1$.
        If $u_1\neq u_n$, then since  $n_1>n_2>0$, we get, $n_1 \geq 2$ which implies $u_1^{n_2} \theta(u') \theta(u_1) u_1^{n_1-1} \in C_\theta(u)\setminus C_\theta(v)$, which is a contradiction. Thus,  $u_1=u_n$ and we get, $u=u_1^{n_1} u' u_1^{n_2}$ and  $v=u_1^{n_2} u' u_1^{n_1}$.       If $\theta(u')\neq u'$, then $u_1^{n_2} \theta(u') u_1^{n_1} \in C_\theta(u)\setminus C_\theta(v)$, which is a contradiction.  Hence, $\theta(u') = u'$. We have the following:
        \begin{itemize}
            \item  If $\theta(a) = a$ for all $a \in \textit{Alph}(u)$, we have,  $\theta(u') = u'^R$. Thus, $u'=u'^R$. This implies as $u=u_1^{n_1} u' u_1^{n_2}$ and  $v=u_1^{n_2} u' u_1^{n_1}$, $u=v^R$.
            \item If there exists some $a\in \textit{Alph}(u)$ such that $\theta(a)\neq a$, then as $\theta(u_1)=u_1$,
    we have, $u_1\neq a$. Now, $\theta(u') = u'$ implies $u'=x b u''\theta(b)\theta(x)$ for $\theta(b)\neq b \in \Sigma$ and $x,u'' \in \Sigma^*$ such that $\theta(d)=d$ for all $d \in \textit{Alph}(x)$.
     Since $\theta(u_1)=u_1$,
    we have, $u_1\neq b$.
    Now,    $u=u_1^{n_1} x b u''\theta(b)\theta(x) u_1^{n_2}$
     and
    $v=u_1^{n_2} xbu''\theta(b)\theta(x) u_1^{n_1}$.
    Consider $w = u_1^{n_2} x b \theta(u'') u_1^{n_1} x b $. Then,  $w \in C_\theta(u)\setminus C_\theta(v)$, which is a contradiction.
        \end{itemize}
            \end{itemize}
    \end{itemize}
 Hence,  $u={u_1}^{n_1} u' {u_1}^{n_2}$ and $v=u^R$ where $u_1 \in \Sigma$, $u' \in \Sigma^+$, $u_1 \notin Pref(u')$, $u_1 \notin Suf(u')$, $u'=u'^R$, $n_1, n_2 \geq 1$, $n_1\neq n_2$, and $\theta(c)=c$ for all $c \in \textit{Alph}(u)$.
\end{proof}

\textbf{A1.} Prove that the language $L_1 = \{  d^n c^k a^n b^k ~|~ n, k \geq 1 \}$ is not context-free by pumping lemma.
\begin{proof}
    Suppose $L_1 = \{  d^n c^k a^n b^k ~|~ n, k \geq 1 \}$ is context-free. Then, by pumping lemma (Theorem \ref{pumth}), there exists a pumping length $p>0$ such that any string $s \in L_1$ with $|s|\geq p$ can be written as $s=uvwxy$, satisfying the following: $|vwx|\leq p$, $|vx|\geq 1$ and $u v^{i} w x^{i} y \in L_1$ for all $i \geq 0$. 
    
    Let $s= d^p c^p a^p b^p$. Clearly, $s \in L_1$ and $|s|> p$. Then pumping lemma guarantees that $s$ can be split into $uvwxy$ where $|vwx|\leq p$, $|vx|\geq 1$ and $u v^{i} w x^{i} y \in L_1$ for all $i \geq 0$.
    Since $|vwx|\leq p$, the factor $vwx$ must be entirely contained with in one of $d^p$, $c^p$, $a^p$, $b^p$, $d^p c^p$, $c^p a^p$ or $a^p b^p$.
    
    If $vwx$ lies entirely in $d^p$, then pumping $v$ and $x$ changes only the number of $d$'s in pumped strings. Then there exists some pumped strings that do not lie in $L_1$. So, $vwx$ does not lie entirely in $d^p$.  
    Similarly, we can show that $vwx$ does not lie entirely in any of $c^p$, $a^p$ and $b^p$. 
    If $vwx$ lies entirely in $d^p c^p$, then pumping $v$ and $x$ changes only the number of $d$'s and $c$'s in pumped strings.  Then there exists some pumped strings that do not lie in $L_1$. So, $vwx$ does not lie entirely in $d^p c^p$.
     Similarly, we can show that $vwx$ does not lie entirely in any of $c^p a^p$ and $a^p b^p$.
     Therefore, in all possible cases, pumping $v$ and $x$ results in some strings that do not lie in $L_1$, which contradicts the pumping lemma. Thus, our assumption that $L_1$ is context-free is false.
\end{proof}

\textbf{A2.} \textit{The proof of Theorem \ref{csl10089}:}

\begin{proof}
{Consider $L$ is a recursively enumerable language. Since the recursively enumerable languages are closed under reversal \cite{rozen1997}, $L^R$ is also recursively enumerable.
Let $\psi : \Sigma^* \rightarrow \Sigma^*$ be a morphic involution such that $\psi(a)=\theta(a)$ for all $a\in \Sigma$.
Consider $$L'=\bigcup_{uv \in L^R} \{ \psi(u)v^R\; \}.$$
Then, by Lemma \ref{mn1},
$$ L'=\bigcup_{uv \in L^R} \{\psi(u)v^R\; \}=  \bigcup_{v^Ru^R \in L} \{ \theta(u^R)v^R \;\} = C_\theta(L).$$\\
We prove that $L'$ is a recursively enumerable language. Since $L^R$ is recursively enumerable,  there exists an unrestricted grammar $G = ( V , \Sigma, P , S )$ such that  $L^R = L(G)$.
We now construct an unrestricted grammar $G_N$ such that $L(G_N)=L'$.
Let $V_1=\{S', T_1, T_2, T_3, T_4, T\}$ and $V_2 = \{ {X_a}^o, {X_a}^I, {X_a}^D, {X_a}^R, {X_a}^T,\\ {X_a}^F, {X_a}^N\;:\; a \in \Sigma  \}$.
Consider $G_N = ( V_N, \Sigma, P_N, S' )$ where $V_N = V \cup V_1 \cup V_2$, $V \cap V_1 = \emptyset$, $V \cap V_2 = \emptyset$, $V_1 \cap V_2 = \emptyset$, $S'$ is the start symbol and $P_N$ is the union of the following sets :
\begin{multicols}{2}
\begin{enumerate}
    \item $P_N(1) = \{S' \rightarrow T_1 S T_2 \;|\; T_1 S T \}$
    \item $P_N(2) = \{ h(\alpha) \rightarrow h(\beta) \;: \;  \alpha \rightarrow \beta \in P \}$
    \item $P_N(3) =\{ T_1 X_{a}^o \rightarrow X_a^N T_1\;:\; a \in \Sigma \}$
    \item $P_N(4) =\{ X_a^o T_2 \rightarrow T_3 X_a^F\;:\; a \in \Sigma \}$
    \item $P_N(5) =\{ X_a^o T_3 \rightarrow T_3 X_a^o\;:\;  a \in \Sigma \}$
    \item $P_N(6) =\{ X_a^o T \rightarrow T_4 X_a^o\;:\; a \in \Sigma \}$
    \item $P_N(7) =\{ X_a^N Z \rightarrow Z \theta(a)\;:  a \in \Sigma  \}$
    \item $P_N(8) =\{ X_a^N Z' \rightarrow Z' \theta(a)\;:\; a \in \Sigma  \}$
    \item $P_N(9) =\{ Z \rightarrow \lambda,  Z'\rightarrow \lambda, Z'' \rightarrow \lambda \}$
     \item $P_N(10) =\{ Z'' X_a^o  \rightarrow X_a^D Z\;:\; a \in \Sigma    \}$
    \item $P_N(11) =\{ X_a^D Z \rightarrow Za\;:\; a \in \Sigma    \}$
    \item $P_N(12) =\{ Z X_a^o \rightarrow Z X_a^I \;:\; a \in \Sigma \}$
    \item $P_N(13) =\{ Z X_{a_1}^I X_{a_2}^o \rightarrow X_{a_1}^I Z X_{a_2}^R\;: a_1, a_2 \in \Sigma\}$
    \item $P_N(14) =\{ Z X_{a_1}^R X_{a_2}^o \rightarrow X_{a_1}^R Z X_{a_2}^R\;:  a_1, a_2 \in \Sigma\}$
    \item $P_N(15) = \{ Z X_{a_1}^R X_{a_2}^F \rightarrow X_{a_2}^T Z X_{a_1}^F\;:  a_1, a_2 \in \Sigma\}$
    \item $P_N(16) =\{  X_{a_1}^R X_{a_2}^T Z \rightarrow X_{a_2}^T Z X_{a_1}^o\;: a_1, a_2 \in \Sigma\}$
    \item $P_N(17) =\{  X_{a_1}^I X_{a_2}^T Z \rightarrow X_{a_2}^D Z X_{a_1}^I\;:a_1, a_2 \in \Sigma\}$ 
    \item $P_N(18) = \{ Z X_{a_1}^I X_{a_2}^F \rightarrow X_{a_2}^D  X_{a_1}^D Z\;: a_1, a_2 \in \Sigma\}$
\end{enumerate}
\end{multicols}
where $Z=T_1T_3, \;Z'=T_1T_2,\; Z''=T_1T_4 $ and $h : ( V \cup \Sigma )^* \rightarrow V_N$  is the morphism defined by $h(\lambda)=\lambda$, 
$h(B) = B$ for all $B \in V$  and $h ( a ) = X_a^o$ for all $a \in \Sigma$.
We prove the following two claims:\\

\textbf{Claim 1 :} $ L' \subseteq L(G_N)$ : Consider $x \in L'$. Then, $x=  \psi(y_1 y_2 \cdots y_i) (y_n y_{n-1}\cdots y_{i+1} )$ where $ y_1 y_2 \cdots y_n \in L^R$ and $y_i \in \Sigma$. This implies as $\psi(a)=\theta(a)$ for all $a\in \Sigma$, $x=  \theta(y_1) \theta(y_2) \cdots \theta(y_i) (y_n y_{n-1}\\  \cdots y_{i+1}) $. Now, we show that $x \in L(G_N)$. Since $ y_1 y_2 \cdots y_n =y \in L^R$, $S$ derives $y$. Then, by $P_N(1)$, $P_N(2)$
and definition of $h$,
$S'$ derives $T_1 X_{y_1}^o X_{y_2}^o \cdots X_{y_n}^o T_2$ or $T_1 X_{y_1}^o X_{y_2}^o \cdots X_{y_n}^o T$. \\
(we use boxes to denote the variables used in each derivation step):\\
Now, 
\begin{align*}
    S'
    \rightarrow^* T_1 X_{y_1}^o X_{y_2}^o \cdots \mybox{$X_{y_n}^o T$}
    &\rightarrow \mybox{$T_1 X_{y_1}^o$} X_{y_2}^o \cdots T_4 X_{y_n}^o & (\text{ by $P_N(6)$ })\\
    &\rightarrow^* {X_{y_1}^N} {X_{y_2}^N} \cdots {X_{y_{n-1}}^N} T_1T_4 X_{y_n}^o & (\text{ by $P_N(3)$ })\\
    &={X_{y_1}^N} {X_{y_2}^N} \cdots {X_{y_{n-1}}^N} \mybox{$Z'' X_{y_n}^o$}\\
    &\rightarrow {X_{y_1}^N} {X_{y_2}^N} \cdots {X_{y_{n-1}}^N} \mybox{$X_{y_n}^D Z$} & (\text{ by $P_N(10)$ })\\
    &\rightarrow^*  \theta({y_1}) \theta({y_2}) \cdots \theta({y_{n-1}})  y_{n} & (\text{ by $P_N(11), P_N(7)$ and $P_N(9)$ })
\end{align*}
So, for $i=n-1$, $x \in L(G_N)$. Now,
\begin{align*}
    S' \rightarrow^* \mybox{$T_1 {X_{y_1}^o}$} {X_{y_2}^o} \cdots {X_{y_n}^o} T_2
    &\rightarrow^* {X_{y_1}^N} \cdots {X_{y_n}^N} T_1 T_2 & (\text{ by $P_N(3)$ })\\
    &= {X_{y_1}^N} \cdots \mybox{${X_{y_n}^N} Z'$}\\
    &\rightarrow^* \theta({y_1}) \cdots \theta({y_n}) & (\text{ by $P_N(8)$ and $P_N(9)$ })
\end{align*}
So, for $i=n$, $x \in L(G_N)$.
Now,
\begin{align*}
    S' &\rightarrow^* \mybox{$T_1 {X_{y_1}^o}$} {X_{y_2}^o} \cdots {X_{y_n}^o} T_2\\
    &\rightarrow  {X_{y_1}^N} \mybox{$T_1 {X_{y_2}^o}$} \cdots {X_{y_n}^o} T_2 \hspace{4.95cm} (\text{by $P_N(3)$}) \\
    &\rightarrow^* {X_{y_1}^N} \cdots {X_{y_i}^N} T_1  {X_{y_{i+1}}^o} \cdots \mybox{${X_{y_n}^o} T_2$} \hspace{3.39cm} (\text{by $P_N(3)$}) \\
    &\rightarrow {X_{y_1}^N} \cdots {X_{y_i}^N} T_1  {X_{y_{i+1}}^o} \cdots \mybox{${X_{y_{n-1}}^o} T_3$} {X_{y_n}^F} \hspace{2.65cm} (\text{by $P_N(4)$})\\
    &\rightarrow^* {X_{y_1}^N} \cdots {X_{y_i}^N} T_1 T_3  {X_{y_{i+1}}^o} \cdots {X_{y_n}^F} \hspace{3.7cm} (\text{by $P_N(5)$}) \\
    &= {X_{y_1}^N} \cdots {X_{y_i}^N} \mybox{$Z  {X_{y_{i+1}}^o}$} {X_{y_{i+2}}^o}\cdots {X_{y_n}^F}  \\
    &\rightarrow {X_{y_1}^N} \cdots {X_{y_i}^N} \mybox{$Z  {X_{y_{i+1}}^I} {X_{y_{i+2}}^o}$}\cdots {X_{y_n}^F}  \hspace{3.18cm} (\text{by $P_N(12)$})\\
    &\rightarrow {X_{y_1}^N} \cdots {X_{y_i}^N}   {X_{y_{i+1}}^I} \mybox{$Z {X_{y_{i+2}}^R} {X_{y_{i+3}}^o}$} \cdots {X_{y_n}^F}  \hspace{2.37cm} (\text{by $P_N(13)$})\\
    &\rightarrow {X_{y_1}^N} \cdots {X_{y_i}^N}   {X_{y_{i+1}}^I}  {X_{y_{i+2}}^R} \mybox{$Z {X_{y_{i+3}}^R} {X_{y_{i+4}}^o}$}\cdots {X_{y_n}^F}  \hspace{1.48cm} (\text{by $P_N(14)$})\\
    &\rightarrow^* {X_{y_1}^N} \cdots {X_{y_i}^N}   {X_{y_{i+1}}^I}  {X_{y_{i+2}}^R}  {X_{y_{i+3}}^R} {X_{y_{i+4}}^R} \cdots \mybox{$Z {X_{y_{n-1}}^R} {X_{y_n}^F}$} \hspace{.42cm} (\text{by $P_N(14)$})\\
    &\rightarrow {X_{y_1}^N} \cdots {X_{y_i}^N}   {X_{y_{i+1}}^I}  {X_{y_{i+2}}^R}   \cdots \mybox{${X_{y_{n-2}}^R} {X_{y_{n}}^T} Z$} {X_{y_{n-1}}^F} \hspace{1.4cm} (\text{by $P_N(15)$})\\
    &\rightarrow {X_{y_1}^N} \cdots {X_{y_i}^N}   {X_{y_{i+1}}^I}  {X_{y_{i+2}}^R}   \cdots \mybox{${X_{y_{n-3}}^R}{X_{y_{n}}^T} Z$} {X_{y_{n-2}}^o}  {X_{y_{n-1}}^F} \hspace{.5cm} (\text{by $P_N(16)$})\\
    &\rightarrow^* {X_{y_1}^N} \cdots {X_{y_i}^N}   \mybox{${X_{y_{i+1}}^I}  {X_{y_{n}}^T} Z$}  {X_{y_{i+2}}^o} \cdots  {X_{y_{n-2}}^o}  {X_{y_{n-1}}^F} \hspace{1.3cm} (\text{by $P_N(16)$})\\
    &\rightarrow {X_{y_1}^N} \cdots {X_{y_i}^N}   {X_{y_{n}}^D} \mybox{$Z {X_{y_{i+1}}^I}   {X_{y_{i+2}}^o}$}   \cdots  {X_{y_{n-2}}^o}  {X_{y_{n-1}}^F} \hspace{1.36cm} (\text{by $P_N(17)$})\\
    \vdots\\
    &\rightarrow^* {X_{y_1}^N} \cdots {X_{y_i}^N}   {X_{y_{n}}^D}  {X_{y_{n-1}}^D}  \cdots  \mybox{$Z {X_{y_{i+1}}^I}  {X_{y_{i+2}}^F}$} \hspace{.9cm} (\text{by prev. derivation})\\
    &\rightarrow {X_{y_1}^N} \cdots {X_{y_i}^N}   {X_{y_{n}}^D}  {X_{y_{n-1}}^D}   \cdots  {X_{y_{i+2}}^D}  \mybox{${X_{y_{i+1}}^D} Z$} \hspace{2.4cm} (\text{by $P_N(18)$})\\
    &\rightarrow {X_{y_1}^N} \cdots {X_{y_i}^N}   {X_{y_{n}}^D}  {X_{y_{n-1}}^D}   \cdots  \mybox{${X_{y_{i+2}}^D}  Z$} y_{i+1}  \hspace{2.54cm} (\text{by $P_N(11)$})\\
    &\rightarrow^* {X_{y_1}^N} \cdots \mybox{${X_{y_i}^N} Z$}  y_{n}  y_{n-1} \cdots  y_{i+2} y_{i+1}  \hspace{3.1cm} (\text{by $P_N(11)$})\\
    &\rightarrow^* \mybox{$Z$} \theta({y_1}) \cdots \theta({y_i})  y_{n}  y_{n-1}  \cdots  y_{i+2} y_{i+1}  \hspace{3cm} (\text{by $P_N(7)$})\\
     &\rightarrow \theta({y_1}) \cdots \theta({y_i})  y_{n}  y_{n-1}  \cdots  y_{i+2} y_{i+1}  \hspace{3.7cm} (\text{by $P_N(9)$})
\end{align*}
Thus, for $i<n-1$, $x \in L(G_N)$. Therefore, $L' \subseteq L(G_N)$.\\
\textbf{Claim 2 : } $L(G_N) \subseteq L' :$ Consider $ y' \in L(G_N)$. We show that $y'$ is in form $ \theta(u^R) v^R$ where $uv \in L^R$.
 Now, $S'$ derives $y'$. Then, by $P_N(1)$ and $P_N(2)$, for some $\alpha \in L^R$, either $S' \rightarrow^* T_1 h(\alpha) T_2 \rightarrow^* y'$ or $S' \rightarrow^* T_1 h(\alpha) T \rightarrow^* y'$.
Consider $\alpha = a_1 a_2 \cdots a_n$ where each $a_j \in \Sigma$.
Then, by definition of $h$, $S'$ derives
\begin{equation}\label{po190}
    T_1 {X_{a_1}^o} {X_{a_2}^o} \cdots {X_{a_n}^o} T_2 \text{ or } T_1 {X_{a_1}^o} {X_{a_2}^o} \cdots {X_{a_n}^o} T.
\end{equation}
\begin{itemize}
\item $S' \rightarrow^* T_1 {X_{a_1}^o} {X_{a_2}^o} \cdots {X_{a_n}^o} T$: Then, it  generates exactly the following word:
\begin{align*}
    S'&\rightarrow^* T_1 {X_{a_1}^o} {X_{a_2}^o} \cdots {X_{a_n}^o} T\\
    &\rightarrow^* {X_{a_1}^N} {X_{a_2}^N} \cdots {X_{a_{n-1}}^N} T_1 T_4 X_{a_n}^o & (\text{by $P_N(3) \& P_N(6)$})\\
    &= {X_{a_1}^N} {X_{a_2}^N} \cdots {X_{a_{n-1}}^N} \mybox{$Z'' X_{a_n}^o$} \\
    &\rightarrow  {X_{a_1}^N} {X_{a_2}^N} \cdots {X_{a_{n-1}}^N}  X_{a_n}^D Z & (\text{by $P_N(10)$})\\
    &\rightarrow^* Z \theta({{a_1}}) \theta({a_2}) \cdots \theta({a_{n-1}}){a_n} & (\text{by $P_N(11) \& P_N(7)$})\\
    &\rightarrow   \theta({{a_1}}) \theta({a_2}) \cdots \theta({a_{n-1}}){a_n} & (\text{by $P_N(9)$})
\end{align*}
Then, $y'= \theta(a_{n-1}\cdots a_2 a_1) a_n \in L'$ where $a_1a_2 \cdots a_n \in L^R$.
\item $S'\rightarrow^* T_1 {X_{a_1}^o} {X_{a_2}^o} \cdots {X_{a_n}^o} T_2 $:  Then, by applying  $P_N(3)$, $P_N(4)$, $P_N(5)$, we have
\begin{align*}
    S' \rightarrow^* T_1 {X_{a_1}^o} {X_{a_2}^o} \cdots {X_{a_n}^o} T_2
     &\rightarrow^*   {X_{a_1}^N} {X_{a_2}^N} \cdots {X_{a_i}^N} T_1 T_3 {X_{a_{i+1}}^o}    \cdots {X_{a_{n-1}}^o} {X_{a_n}^F}\\
     &=  {X_{a_1}^N} {X_{a_2}^N} \cdots {X_{a_i}^N} Z {X_{a_{i+1}}^o}    \cdots {X_{a_{n-1}}^o} {X_{a_n}^F} 
\end{align*}
Now, we have the following cases:
\begin{itemize}
    \item $i=n $: Then, 
            \begin{align*} 
            S'\rightarrow^* {X_{a_1}^N} {X_{a_2}^N} \cdots {X_{a_n}^N} T_1T_2
            &= {X_{a_1}^N} {X_{a_2}^N} \cdots {X_{a_n}^N} Z'\\
            &\rightarrow^* Z' \theta({a_1}) \theta({a_2}) \cdots \theta({a_n}) & (\text{by $P_N(8)$})\\
            & \rightarrow \theta({a_1}) \theta({a_2}) \cdots \theta({a_n}) & (\text{by $P_N(9)$})
            \end{align*}
            Then, $y'= \theta( a_n \cdots a_2 a_1) \in L'$ where $a_1a_2 \cdots a_n \in L^R$.
    \item $i=n-1$: Then, as $Z=T_1 T_3$,
    $$S' \rightarrow^* {X_{a_1}^N} {X_{a_2}^N} \cdots {X_{a_{n-1}}^N} Z {X_{a_n}^F}$$
    which can not derive a sequence of terminals.
    \item $i < n-1$: Then, $$S' \rightarrow^* {X_{a_1}^N} {X_{a_2}^N} \cdots {X_{a_i}^N} Z {X_{a_{i+1}}^o}    \cdots {X_{a_{n-1}}^o} {X_{a_n}^F}.$$
    Now, we can apply one of $P_N(7)$, $P_N(9)$ and $P_N(12)$. If we apply $P_N(7)$ or $P_N(9)$, then we can not get a string of terminals. Thus, by applying $P_N(12)$, we have
    \begin{align*}
        S' 
        &\rightarrow^* {X_{a_1}^N} {X_{a_2}^N} \cdots {X_{a_i}^N} Z {X_{a_{i+1}}^I} {X_{a_{i+2}}^o}   \cdots {X_{a_{n-1}}^o} {X_{a_n}^F}
    \end{align*}
    Again, if we apply $P_N(7)$ or $P_N(9)$, then we can not get a string of terminals. So, we have to apply $P_N(13)$. Let $|{X_{a_{i+1}}^I} {X_{a_{i+2}}^o}   \cdots {X_{a_{n-1}}^o} {X_{a_n}^F}|=m\geq 2$.
    We now use induction on $m$ to prove that $$Z{X_{a_{i+1}}^I} {X_{a_{i+2}}^o}   \cdots {X_{a_{n-1}}^o} {X_{a_n}^F} \rightarrow^* Z a_n a_{n-1} \cdots a_{i+1}.$$ 
    For $m=2$, $$ Z {X_{a_{n-1}}^I} {X_{a_n}^F} \rightarrow X_{a_n}^D {X_{a_{n-1}}^D} Z
    \rightarrow^* Z a_n a_{n-1}.$$
    For $m=3$, 
    \begin{align*}
        \mybox{$Z {X_{a_{n-2}}^I} {X_{a_{n-1}}^o}$} {X_{a_n}^F}
        &\rightarrow {X_{a_{n-2}}^I} \mybox{$Z {X_{a_{n-1}}^R} {X_{a_n}^F}$} & (\text{by $P_N(13)$}) \\
        &\rightarrow \mybox{${X_{a_{n-2}}^I}  {X_{a_{n}}^T} Z$} {X_{a_{n-1}}^F}   & (\text{by $P_N(15)$})\\
        &\rightarrow {X_{a_{n}}^D} \mybox{$Z {X_{a_{n-2}}^I}  {X_{a_{n-1}}^F}$}  & (\text{by $P_N(17)$})\\
    &\rightarrow {X_{a_{n}}^D}  {X_{a_{n-1}}^D}  {X_{a_{n-2}}^D} Z  & (\text{by $P_N(18)$})\\
    &\rightarrow^* Z a_n a_{n-1} a_{n-2}  & (\text{by $P_N(11)$})
    \end{align*}
    Let the statement be true for strings with lengths less than $m$. Now,
\begin{align*}
    &\mybox{$Z{X_{a_{i+1}}^I} {X_{a_{i+2}}^o}$}   \cdots {X_{a_{n-1}}^o} {X_{a_n}^F}\\
    &\rightarrow {X_{a_{i+1}}^I} \mybox{$Z {X_{a_{i+2}}^R} {X_{a_{i+3}}^o}$} \cdots {X_{a_{n-1}}^o} {X_{a_n}^F}  & (\text{by $P_N(13)$})\\
    &\rightarrow {X_{a_{i+1}}^I}  {X_{a_{i+2}}^R} Z {X_{a_{i+3}}^R} \cdots {X_{a_{n-1}}^o} {X_{a_n}^F}  & (\text{by $P_N(14)$})\\
    &\rightarrow^* {X_{a_{i+1}}^I}  {X_{a_{i+2}}^R} {X_{a_{i+3}}^R} \cdots \mybox{$Z {X_{a_{n-1}}^R} {X_{a_n}^F}$}\\
    &\rightarrow {X_{a_{i+1}}^I}  {X_{a_{i+2}}^R} {X_{a_{i+3}}^R} \cdots \mybox{${X_{a_{n-2}}^R} {X_{a_n}^T} Z$} {X_{a_{n-1}}^F}  & (\text{by $P_N(15)$}) \\
    &\rightarrow {X_{a_{i+1}}^I}  {X_{a_{i+2}}^R} {X_{a_{i+3}}^R} \cdots {X_{a_{n-3}}^R} {X_{a_n}^T} Z {X_{a_{n-2}}^o} {X_{a_{n-1}}^F}  & (\text{by $P_N(16)$})
\end{align*}
Here, we can apply one of $P_N(9)$, $P_N(12)$ and $P_N(16)$. If we apply $P_N(9)$ or $P_N(12)$, then $S'$ does not derive a string of terminals. Thus, we apply $P_N(16)$. Then,
\begin{align*}
    Z{X_{a_{i+1}}^I} {X_{a_{i+2}}^o}   \cdots {X_{a_{n-1}}^o} {X_{a_n}^F}
    &\rightarrow^* {X_{a_{i+1}}^I}  {X_{a_{i+2}}^R} {X_{a_{i+3}}^R} \cdots {X_{a_{n-4}}^R} {X_{a_n}^T} Z {X_{a_{n-3}}^o} {X_{a_{n-2}}^o} {X_{a_{n-1}}^F}
\end{align*}    
Now, similar to the above, we can not apply $P_N(9)$ or $P_N(12)$ here. So, we apply  $P_N(16)$ and continue the process. Then, we have    
    \begin{align*}
    Z{X_{a_{i+1}}^I} {X_{a_{i+2}}^o}   \cdots {X_{a_{n-1}}^o} {X_{a_n}^F}
    &\rightarrow^* \mybox{${X_{a_{i+1}}^I} {X_{a_{n}}^T} Z$} {X_{a_{i+2}}^o}    \cdots  {X_{a_{n-2}}^o}  {X_{a_{n-1}}^F}\\
    &\rightarrow  {X_{a_{n}}^D} Z {X_{a_{i+1}}^I} {X_{a_{i+2}}^o}    \cdots  {X_{a_{n-2}}^o}  {X_{a_{n-1}}^F} & (\text{by $P_N(17)$})\\
    &\rightarrow^* {X_{a_{n}}^D} Z a_{n-1} \cdots a_{i+1} & (\text{by induction})\\
    &\rightarrow Z a_n a_{n-1} \cdots a_{i+1} & (\text{by $P_N(11)$})
\end{align*} 
Therefore,  \begin{align*}
        S' &\rightarrow^* {X_{a_1}^N} {X_{a_2}^N} \cdots {X_{a_i}^N} Z {X_{a_{i+1}}^I} {X_{a_{i+2}}^o}   \cdots {X_{a_{n-1}}^o} {X_{a_n}^F}\\
        &\rightarrow^* {X_{a_1}^N} {X_{a_2}^N} \cdots {X_{a_i}^N} Z a_n a_{n-1} \cdots a_{i+1}\\
        &\rightarrow^*  \theta(a_1) \theta(a_2) \cdots \theta(a_i) a_n a_{n-1} \dots a_{i+1}  & (\text{by $P_N(7)$ and $P_N(9)$}).
    \end{align*}  
Then, $y'= \theta(a_i a_{i-1} \cdots a_2 a_1) a_n a_{n-1} \dots a_{i+1} \in L' $ where $a_1a_2\cdots a_n \in L^R$.    
\end{itemize}
\end{itemize}
Therefore, $L(G_N) \subseteq  L'$.}
\end{proof}

\textbf{A3. }
A language $L$ is conjugate-free if for all $w\in L$, $C(w) \cap L = \{w\}$.
It is undecidable to determine whether or not a given context-free
language $L$ is conjugate-free.

\begin{proof}
Let $\Sigma$ be alphabet and $(U, V)$ be an instance of PCP problem where $U=(u_0, u_1, \ldots, u_{n-1})$, $V=(v_0, v_1, \ldots, v_{n-1})$ with each $u_i, v_i \in \Sigma^*$. 
Assume that the symbols $X, Y, a_0, a_1, \cdots, a_{n-1}$ are distinct and not in $\Sigma$. Consider $\Sigma'=\Sigma \cup \{ X, Y, a_0, a_1, \cdots, a_{n-1}\}$.
We define a context-free grammar $G = (N, \Sigma', P, S)$ where $N = \{S, A, B\}$ is the set of non-terminals, $\Sigma'$ is the set of terminals, $S$ is the start symbol, and $P$ is the set of derivation rules which contains the following rules:
\begin{itemize}
    \item $ S \rightarrow XA ~ | ~ YB  $
    \item $A \rightarrow u_0 A a_0 ~ |~u_1 A a_1 ~ |~ \cdots~|~  u_{n-1} A a_{n-1} ~ | ~ Y$
    \item $B \rightarrow a_0 B v_0 ~|~ a_1 B v_1 ~|~ \cdots ~|~ a_{n-1} B v_{n-1} ~|~ X $
\end{itemize}

Now, $L(G)$ denotes the set of all words generated by the grammar $G$ and any string of $L(G)$ is in one of the following forms :
\begin{equation*}
   X u_{j_{1}} u_{j_{2}} \cdots u_{j_{l}} Y   a_{j_l} \cdots  a_{j_2} a_{j_1} 
   \tag{$i$}
\end{equation*}
or
\begin{equation*}
    Y  a_{j_1}  a_{j_2} \cdots  a_{j_l} X v_{j_l} \cdots v_{j_2} v_{j_1}.
     \tag{$ii$}
\end{equation*}

We now show that $L(G)$ is not conjugate-free if and only if the PCP instance $(U, V)$ has a solution.
First, we prove that if the PCP instance $(U, V)$ has a solution, then $L(G)$ is not conjugate-free.
Assume that the PCP instance $(U, V)$ has a solution. 
Then, there is a sequence of integers $ i_0, i_1, \ldots, i_k$ such that
$u_{i_0} u_{i_1} \cdots u_{i_k} = v_{i_0} v_{i_1} \cdots v_{i_k}$.
Now, $L(G)$ contains the following elements:
\begin{align*}
       &\beta' =   Y a_{i_k}  a_{i_{k-1}} \cdots  a_{i_0} X v_{i_0} \cdots v_{i_{k-1}} v_{i_k}  
   = Y  a_{i_k}  a_{i_{k-1}} \cdots  a_{i_0} X u_{i_0} u_{i_{1}} \cdots   u_{i_{k}} ~\text{ and } \\
   &\alpha' = X  u_{i_0} u_{i_{1}} \cdots u_{i_{k}} Y  a_{i_k} \cdots a_{i_1} a_{i_0}. 
   \end{align*}
Clearly, $\beta' \in C(\alpha')$. Since $X \neq Y$,  $\alpha' \neq \beta'$.
Thus,  $L(G)$ is not conjugate-free.

Now, we prove the other part, i.e., if $L(G)$ is not conjugate-free, then the PCP instance $(U, V)$ has a solution. 
Assume that $L(G)$ is not conjugate-free, i.e., there exist $w_1, w_2$ in $L(G)$ such that $w_2 \in C(w_1)$ and $w _1 \neq w_2$. 
Then, $w_1=xy, w_2=y x$  for some $ x, y \in {\Sigma'}^+$.
Since elements of $L(G)$ are in the form $(i)$ or $(ii)$, we have two cases:
\begin{itemize}
    \item \textbf{Case I :} Consider $w_1=xy$ is in the form $(i)$. Then, for some integer $m$, 
    $$w_1=xy= X u_{i_{1}} u_{i_{2}} \cdots u_{i_{m}} Y   a_{i_m} \cdots  a_{i_2} a_{i_1}.$$ 
    
 Now, $\{X, Y, a_0, a_1, \cdots, a_{n-1}\} \cap \Sigma = \emptyset$, each $u_i \in \Sigma^*$, and $x, y \in {\Sigma'}^+$. Then as $w_2 \in C(w_1)$, $w_2$ must be in the form $(ii)$.
   Then, for some integer $r$,
    $$ w_2 = yx = Y  a_{j_r}  a_{j_{r-1}} \cdots  a_{j_1} X v_{j_1} \cdots v_{j_{r-1}} v_{j_r}. $$
   Since $ |w_1|_{Y} =1$, $Y \in Pref(w_2)$ and $w_2 \in C(w_1)$, we have $x=X u_{i_{1}} u_{i_{2}} \cdots u_{i_{m}}$ and $y = Y   a_{i_m} \cdots  a_{i_2} a_{i_1}$.
    Then, $ w_2 = yx = Y  a_{j_r}  a_{j_{r-1}} \cdots  a_{j_1} X v_{j_1} \cdots v_{j_{r-1}} v_{j_r}$ implies
    \begin{equation}\label{erq12}
        Y   a_{i_m} \cdots  a_{i_2} a_{i_1}  X u_{i_{1}} u_{i_{2}} \cdots u_{i_{m}} = Y  a_{j_r}  a_{j_{r-1}} \cdots  a_{j_1} X v_{j_1} \cdots v_{j_{r-1}} v_{j_r}
    \end{equation}
      Using our assumed conditions, we have the following from Equation (\ref{erq12}):   $r=m$, ${i_1}={j_1}$, ${i_2}={j_2}$, \ldots, ${i_m}={j_m}$ and 
        $u_{i_1} u_{i_2} \cdots u_{i_m} = v_{i_1} v_{i_2} \cdots v_{i_m} $. Hence, the PCP instance has a solution.
    \item \textbf{Case II :} Consider $w_1=xy$ is in the form $(ii)$. 
    Then similar to Case I, we show that the PCP instance has a solution.  
\end{itemize}
Therefore, as PCP is undecidable \cite{post1946}, determining whether or not $L$ is conjugate-free is also undecidable when $L$ is context-free.

\end{proof}

\end{document}